\pdfoutput=1

\newif\iffull
\fulltrue%
\iffull
\documentclass[acmsmall,screen,authorversion]{acmart}
\else
\documentclass[acmsmall,screen]{acmart}
\fi

\usepackage[T1]{fontenc}
\usepackage[utf8]{inputenc} 
\usepackage{upgreek}
\usepackage{paralist}
\usepackage{enumitem}
\usepackage{url}
\usepackage{stmaryrd}
\usepackage{relsize}
\usepackage{graphicx}
\usepackage{todonotes}
\usepackage{xspace}

\usepackage{amsmath}
\usepackage{amssymb}
\usepackage{amsthm}
\usepackage{stmaryrd}
\usepackage{mathtools}
\usepackage{color}
\usepackage{mathpartir}
\usepackage{microtype}
\usepackage[capitalise]{cleveref}
\usepackage{subcaption}
\usepackage{listings}
\usepackage{thmtools}
\usepackage{thm-restate}
\usepackage{wrapfig}

\newtheorem*{theorem*}{Theorem}

\newtheorem{lemma}{Lemma}
\newtheorem*{lemma*}{Lemma}
\newtheorem*{remark*}{Remark}
\newtheorem{example}{Example}
\newtheorem{definition}{Definition}

\newcommand{\fullref}[1]{\iffull\cref{#1}\else the full version~\citep{BHL19}\fi}

\newif\ifcomments
\commentsfalse%
\ifcomments\overfullrule=2cm\fi

\newcommand\sepid{\ensuremath{I}}
\newcommand\sepand{\mathrel{*}}
\newcommand\sepimp{\mathrel{-\mkern-6mu*}}
\newcommand{\bigsep}{\mathop{\Huge \mathlarger{\mathlarger{*}}}}

\newcommand{\psl}[3]{\{ #1 \}\ #2\ \{ #3 \}}

\newcommand{\Skip}{\mathbf{skip}}
\newcommand{\Seq}[2]{{#1} \mathrel{;} {#2}}
\newcommand{\Assn}[2]{\ensuremath{{#1} \leftarrow {#2}}}
\newcommand{\Rand}[2]{{#1} \stackrel{\raisebox{-.25ex}[.25ex]%
{\tiny $\mathdollar$}}{\raisebox{-.2ex}[.2ex]{$\leftarrow$}} {#2}}
\newcommand{\DCond}[3]{\mathbf{if}_D\ #1\ \mathbf{then}\ #2\ \mathbf{else}\ #3}
\newcommand{\DCondt}[2]{\mathbf{if}_D\ #1\ \mathbf{then}\ #2}
\newcommand{\DFor}[4]{\mathbf{for}\ #1 = #2, \dots, #3\ \mathbf{do}\ #4}
\newcommand{\DWhile}[2]{\mathbf{while}\ #1\ \mathbf{do}\ #2}
\newcommand{\RCond}[3]{\mathbf{if}_R\ #1\ \mathbf{then}\ #2\ \mathbf{else}\ #3}
\newcommand{\RCondt}[2]{\mathbf{if}_R\ #1\ \mathbf{then}\ #2}

\newcommand{\lcp}{\mathsf{lcp}}
\newcommand{\iread}[1]{\ensuremath{\mathsf{read}(#1)}}
\newcommand{\iwrite}[2]{\ensuremath{\mathsf{write}(#1,#2)}}

\newcommand{\denot}[1]{\ensuremath{\llbracket #1 \rrbracket}}
\newcommand{\dom}{\text{dom}}

\newcommand{\Val}{\ensuremath{\mathbf{Val}}}

\newcommand{\DVar}{\ensuremath{\mathcal{DV}}}
\newcommand{\RVar}{\ensuremath{\mathcal{RV}}}
\newcommand{\DMem}{\ensuremath{\mathbf{DetM}}}
\newcommand{\RMem}{\ensuremath{\mathbf{RanM}}}
\newcommand{\DExp}{\ensuremath{\mathcal{DE}}}
\newcommand{\RExp}{\ensuremath{\mathcal{RE}}}
\newcommand{\Com}{\ensuremath{\mathcal{C}}}

\newcommand{\RCom}{\ensuremath{\mathcal{RC}}}
\newcommand{\dbind}{\text{bind}}
\newcommand{\dunit}{\text{unit}}
\newcommand{\dcond}[2]{\ensuremath{{#1} \mid {#2}}}
\newcommand{\dconv}[3]{\ensuremath{{#2} \oplus_{#1} {#3}}}
\newcommand{\true}{\text{true}}
\newcommand{\false}{\text{false}}
\newcommand{\ktt}{\ensuremath{\mathit{tt}}}
\newcommand{\kff}{\ensuremath{\mathit{ff}}}
\newcommand{\APred}{\ensuremath{\mathcal{AP}}}
\newcommand{\CM}{\ensuremath{\text{CM}}}
\newcommand{\SP}{\ensuremath{\text{SP}}}
\newcommand{\FV}{\ensuremath{\mathit{FV}}}
\newcommand{\MV}{\ensuremath{\mathit{MV}}}
\newcommand{\RV}{\ensuremath{\mathit{RV}}}
\newcommand{\WV}{\ensuremath{\mathit{WV}}}

\newcommand{\supp}{\ensuremath{\text{supp}}}
\newcommand{\Dist}{\ensuremath{\mathbf{D}}}

\newcommand{\Unif}{\ensuremath{\mathbf{U}}}

\newcommand{\PWHILE}{\textsc{pWhile}}
\newcommand{\ELLORA}{\textsc{Ellora}}
\newcommand{\QSL}{\textsc{QSL}}
\newcommand{\PPDL}{\textsc{PPDL}}
\newcommand{\PGCL}{\textsc{pGCL}}
\newcommand{\SYSTEM}{\textsc{PSL}}
\newcommand{\rname}[1]{\textsc{#1}}

\let\origthelstnumber\thelstnumber%
\makeatletter
\newcommand*\Suppressnumber{%
  \lst@AddToHook{OnNewLine}{%
    \let\thelstnumber\relax%
     \advance\c@lstnumber-\@ne\relax%
    }%
}
\newcommand*\Reactivatenumber{%
  \lst@AddToHook{OnNewLine}{%
   \let\thelstnumber\origthelstnumber%
   \advance\c@lstnumber\@ne\relax}%
}
\makeatother

\setcopyright{rightsretained}
\acmPrice{}
\acmDOI{10.1145/3371123}
\acmYear{2020}
\copyrightyear{2020}
\acmJournal{PACMPL}
\acmVolume{4}
\acmNumber{POPL}
\acmArticle{55}
\acmMonth{1}

\begin{document}

\title{A Probabilistic Separation Logic}


\author{Gilles Barthe}
\affiliation{%
 \institution{MPI for Security and Privacy, Germany and IMDEA Software
 Institute, Spain}}

\author{Justin Hsu}
\affiliation{%
 \institution{University of Wisconsin--Madison, USA}}

\author{Kevin Liao}
\affiliation{%
 \institution{MPI for Security and Privacy, Germany and University of Illinois
 Urbana-Champaign, USA}}


\begin{abstract}
  \emph{Probabilistic independence} is a useful concept for describing the
  result of random sampling---a basic operation in all probabilistic
  languages---and for reasoning about groups of random variables.  Nevertheless,
  existing verification methods handle independence poorly, if at all. We
  propose a probabilistic separation logic \SYSTEM, where separation models
  probabilistic independence. We first give a new, probabilistic model of the
  logic of bunched implications (BI). We then build a program logic based on
  these assertions, and prove soundness of the proof system. We demonstrate our
  logic by verifying information-theoretic security of cryptographic
  constructions for several well-known tasks, including private information
  retrieval, oblivious transfer, secure multi-party addition, and simple
  oblivious RAM. Our proofs reason purely in terms of high-level properties,
  like independence and uniformity.
\end{abstract}



\maketitle

\section{Introduction}%
\label{sec:intro}


Probabilistic programs have important applications in many domains, including
information security and machine learning. As the impact of these areas
continues to grow, probabilistic programming languages (PPLs) are receiving
renewed attention from formal verification. While the mathematical semantics of
PPLs has been well-studied, starting from
\citet{Kozen81,DBLP:journals/tcs/Saheb-Djahromi80} and continuing up to
today~\citep{DBLP:journals/pacmpl/EhrhardPT18,DBLP:journals/pacmpl/VakarKS19},
deductive program verification for PPLs remains challenging. Establishing simple
properties can involve tedious arguments, and scaling formal proofs up to verify
target properties of randomized algorithms is often difficult.

\subsection{Probabilistic Independence}
A basic property that is poorly handled by existing verification techniques is
\emph{independence}. Roughly speaking, two random variables are
probabilistically independent if they are uncorrelated: information about one
quantity yields no information about the other. In probabilistic programs,
independence usually arises when variables are derived from \emph{separate
randomness}---e.g., from the results of two different coin flips---but
independence can also hold when variables share randomness in just the right
way.

Although it is usually not the target property of interest, probabilistic
independence often serves as an intermediate assertion in pen-and-paper proofs
of randomized algorithms. From a verification perspective, independence is
useful for several reasons.

\paragraph*{Independence simplifies reasoning about groups of random variables.}
Probabilistic programs often manipulate multiple random variables. If a group of
random variables are independent, then their joint distribution is precisely
described by the distribution of each variable in isolation. As a result, formal
reasoning can focus on one variable at a time, without losing information.

\paragraph*{Independence characterizes the result of random sampling.}
All PPLs have built-in constructs to draw random samples from primitive
distributions (e.g., drawing a random boolean from a coin-flip distribution).
These basic operations produce a ``fresh'' random quantity that is
\emph{independent} from the rest of the program state, at least when the
primitive distribution does not depend on the state.

\paragraph*{Independence is preserved under local operations.}
Like standard programs, probabilistic programs typically manipulate only a few
variables at a time. To ease formal reasoning, properties about unmodified
variables should be preserved as much as possible. Independence is preserved
under local modifications: if $x$ and $y$ are independent and $x$ is updated to
$x' = x + 1$, then $x'$ and $y$ remain independent. In this way, probabilistic
independence seemingly flows through a program, continuing to hold far beyond
the original sampling instructions.

\paragraph*{Independence is compatible with conditioning.}
Probabilistic programs can have randomized control flow, for example branching
on a randomized boolean. Semantically, this kind of branch is modeled by
\emph{conditioning}, an operation that transforms an input distribution into two
\emph{conditional distributions}, one where the guard is true and one where the
guard is false.

While conditioning is well-understood mathematically, it poses problems for
formal reasoning. Conditioning on a variable $x$---say, when branching on $x >
0$---can alter the distribution over other variables. If $x$ is independent of
$y$, however, conditioning on $x$ will have no effect on the distribution over
$y$. As a result, properties of variables that are \emph{syntactically} separate
from the variable $x$ are \emph{preserved} when conditioning on the guard $x >
0$, a highly useful reasoning principle.

\subsection{Example Applications of Independence}

To ground our investigation in applications, we focus on security properties
from cryptography. We first encode target cryptographic protocols as
probabilistic programs of type $\mathcal{A} \times \mathcal{A} \rightarrow
\mathcal{O} \times \mathcal{A}$, where the first input represents the secret
input, the second input represents the public input, and the first and second
output represents the observer's view and computation output, respectively. In
many cases the two outputs coincide, but this need not be the case in general. 

Then, we establish security properties by proving properties of these programs
in our logic. Defining precisely what it means for a construction to be secure
is surprisingly subtle; cryptographers have proposed many definitions capturing
different assumptions and guarantees. Baseline, information-theoretic security
of many schemes, including private information
retrieval~\citep{chor1995private}, oblivious
transfer~\citep{rivest1999unconditionally}, multi-party
computation~\citep{cramer2015secure}, and oblivious RAM~\citep{iacr/ChungP13a},
can be stated in terms of the following definitions.

\paragraph*{Uniformity.}
A natural way to define security is to require that the observer's view is the
same, no matter what the private input is; this is a probabilistic form of
non-interference. For instance, it suffices to show that the observer's view is
always \emph{uniformly distributed} over a fixed set: no matter what the private
inputs are, the observer's view is the same.

\paragraph*{Input Independence.}
Another way to define security is to model the secret input as drawn from some
distribution, and then argue that the distribution of the observer's view is
\emph{probabilistically independent} of the secret input. This formulation
captures security through an intuitive reading of independence: the observer's
view reveals no information about the secret input. Though this definition looks
quite different from probabilistic non-interference, the two definitions are
equivalent in many settings. Their proofs, however, may be quite different.

\subsection{Contributions and Plan of the Paper}

After introducing mathematical preliminaries in \cref{sec:prelim}, we begin
working towards the main goal of this paper: a probabilistic program logic where
independence is the central concept. Our logic is a probabilistic variant of
separation logic, a highly successful technique for reasoning about
heap-manipulating programs~\citep{OhRY01,IOh01}. To model sharing and
separation, separation logic uses assertions from the logic of bunched
implications (BI), a substructural logic. For instance, the \emph{separating
conjunction} models separation of heaps: $\phi \sepand \psi$ states that the
heap can be split into two \emph{disjoint} parts satisfying $\phi$ and $\psi$,
respectively.

While separation logic was originally designed for heaps, separation is a useful
concept in many verification settings. A notable line of work extends separation
logic to the concurrent setting, where separation models exclusive ownership of
resources~\citep{DBLP:journals/tcs/OHearn07,DBLP:journals/tcs/Brookes07}.  More
generally, the \emph{resource semantics} of BI~\citep{DBLP:journals/tcs/PymOY04}
gives a powerful way to generalize BI to new notions of separation.

Inspired by this perspective, our \emph{first contribution} is a new
interpretation of BI where the separating conjunction models probabilistic
independence. Roughly speaking, $\phi \sepand \psi$ holds in a distribution
$\mu$ over program memories if $\mu$ can be factored into two distributions
$\mu_1$ and $\mu_2$ satisfying $\phi$ and $\psi$, respectively. Splitting a
distribution amounts to finding two disjoint sets of program variables $X$ and
$Y$ such that every distribution in the support of $\mu$ is defined precisely on
$X \cup Y$, with the factors $\mu_1$ and $\mu_2$ obtained by projecting $\mu$
along $X$ and $Y$ respectively. This intuitive interpretation gives rise to a
probabilistic model of BI. Our model can smoothly incorporate useful primitive
assertions about distributions, including probabilistic equality and uniformity.
We present our model in \cref{sec:pbi}.

Leveraging this probabilistic version of BI as an assertion logic, our
\emph{second contribution} is a program logic \SYSTEM{} for a simple
probabilistic programming language, similar to \PWHILE\@. Our logic bears a
strong resemblance to separation logic: there are proof rules for local and
global reasoning, there is a version of the Frame rule, and whereas separation
logic distinguishes between \emph{store} and \emph{heap}, our logic
distinguishes between \emph{deterministic} and \emph{probabilistic} variables.
However, there are also notable differences in the probabilistic setting. We
present the proof system of \SYSTEM{} and prove soundness in \cref{sec:psl}.

As our \emph{third contribution}, we demonstrate our program logic by
formalizing security of several well-known constructions from cryptography,
including a simple oblivious RAM, a private information retrieval algorithm, a
simple three-party computation algorithm for addition, and an oblivious transfer
algorithm. We prove two different forms of information-theoretic security:
uniformity of outputs (which implies probabilistic non-interference), and input
independence. We present these examples in \cref{sec:examples}.

We survey related work in \cref{sec:rw} and discuss potential future directions
in \cref{sec:conclusion}.

\section{Preliminaries}%
\label{sec:prelim}

\subsection{Probabilities and Distributions}

A \emph{(discrete) probability distribution} over a countable set $A$ is a
function $\mu : A \to [0, 1]$ such that the total weight is one: $\sum_{a \in A}
\mu(a) = 1$; we write $\Dist(A)$ for the set of all distributions over $A$.
Intuitively, $\mu(a)$ represents the \emph{probability} of drawing $a$ from the
distribution $\mu$. Likewise, the probability of drawing some element in $S
\subseteq A$ is $\mu(S) \triangleq \sum_{a \in S} \mu(a)$. The \emph{support} of
a distribution is the set of elements with non-zero probability: $\supp(\mu) =
\{ a \in A \mid \mu(a) > 0 \}$. 

We will use two standard constructions on probability distributions. First, the
\emph{distribution unit} $\dunit : A \to \Dist(A)$ associates each element $a
\in A$ with the \emph{Dirac distribution} $\delta_a$ centered at $a$. This
distribution is simply defined as $\delta_a(x) \triangleq 1$ if $x = a$ and
$\delta_a(x) \triangleq 0$ otherwise; intuitively, the Dirac distribution
deterministically yields $a$. Second, the \emph{distribution bind}
$\dbind : \Dist(A) \to (A \to \Dist(B)) \to \Dist(B)$ is defined by:
\[
  \dbind(\mu, f)(b) \triangleq \sum_{a \in A} \mu(a) \cdot f(a)(b)
\]
Intuitively, $\dbind$ sequences a distribution with a continuation. Together,
$\dunit$ and $\dbind$ make $\Dist$ a monad~\citep{10.1007/BFb0092872}; these
operations are commonly used to model randomized programs.

Given our focus on independence, we will be particularly interested in
distributions over products and products of distributions. The distribution
\emph{product} $\otimes : \Dist(A) \times \Dist(B) \to \Dist(A \times B)$ is
defined by:
\[
  (\mu_A \otimes \mu_B)(a, b) \triangleq \mu_A(a) \cdot \mu_B(b) .
\]
We can extract component distributions out of any distribution over a product
using the projections $\pi_1 : \Dist(A \times B) \to \Dist(A)$ and $\pi_2 :
\Dist(A \times B) \to \Dist(B)$:
\[
  \pi_1(\mu)(a) \triangleq \sum_{b \in B} \mu(a, b)
  \quad\text{and}\quad
  \pi_2(\mu)(b) \triangleq \sum_{a \in A} \mu(a, b) .
\]
We call $\mu \in \Dist(A \times B)$ a \emph{product distribution} if it can be
factored as $\mu = \pi_1(\mu) \otimes \pi_2(\mu)$; in this case, we say that the
components of $\mu$ are \emph{(probabilistically) independent}.

Finally, will need conditioning and convex combination operations on
distributions to model control flow splits and merges, respectively. Let $S
\subseteq A$ be any event. If $S$ has non-zero probability under $\mu \in
\Dist(A)$, then the \emph{conditional distribution} $\dcond{\mu}{S} \in
\Dist(A)$ is defined as:
\[
  (\dcond{\mu}{S})(E) \triangleq \frac{\mu(S \cap E)}{\mu(S)} .
\]
Intuitively, the conditional distribution represents the relative probabilities
of elements restricted to $S$. Conditioning is not defined when $\mu(S) = 0$.

To join output distributions from two branches, we define the convex combination
of distributions. Let $\rho \in [0, 1]$ and let $\mu_1, \mu_2 \in \Dist(A)$. The
\emph{convex combination} $\dconv{\rho}{\mu_1}{\mu_2} \in \Dist(A)$ is defined
as:
\[
  (\dconv{\rho}{\mu_1}{\mu_2})(S) \triangleq \rho \cdot \mu_1(S) + (1 - \rho) \cdot \mu_2(S) .
\]
We define $\dconv{0}{\mu_1}{\mu_2} \triangleq \mu_2$ and
$\dconv{1}{\mu_1}{\mu_2} \triangleq \mu_1$, even when $\mu_1$ or $\mu_2$ may be
undefined. Conditioning and taking convex combination yields the original
distribution: $\mu =
\dconv{\mu(S)}{(\dcond{\mu}{S})}{(\dcond{\mu}{\overline{S}})}$.

\subsection{Probabilistic Memories}

Distributions over program memories are naturally modeled by distributions over
products. We fix a countable set $\RVar$ of \emph{random variables} and a
countable set $\Val$ of values. For any subset of variables $S \subseteq \RVar$,
we let $\RMem[S] \triangleq S \to \Val$ be the set of memories with domain $S$;
we write $\RMem \triangleq \RMem[\RVar]$.  When $S$ is empty, there is precisely
one map $0 : \emptyset \to \Val$ and so $\Dist(\RMem[\emptyset])$ contains just
the Dirac distribution $\delta_0$. Given a distribution $\mu \in
\Dist(\RMem[S])$, we write $\dom(\mu) \triangleq S$ for the domain. 

Viewing $\RMem[S]$ as a product indexed by $S$, we can adapt the general
constructions for distributions over products to distributions over $\RMem[S]$.
Given disjoint variables $S, S' \subseteq \RVar$, for instance, we define the
product $\otimes : \Dist(\RMem[S]) \times \Dist(\RMem[S']) \to \Dist(\RMem[S
\cup S'])$ to be
\[
  (\mu_S \otimes \mu_{S'})(m) \triangleq \mu_S(m_S) \cdot \mu_{S'}(m_{S'})
\]
where $m_S \in \RMem[S]$ and $m_{S'} \in \RMem[S']$ restrict $m$ to $S$ and $S'$
respectively. The Dirac distribution $\delta_0$ is the identity of this
operation: $\mu \otimes \delta_0 = \delta_0 \otimes \mu = \mu$. When $\mu \in
\Dist(\RMem[S])$ can be factored as $\mu = \mu_1 \otimes \mu_2$ for $\mu_i \in
\Dist(\RMem[S_i])$, we say that $S_1$ and $S_2$ are \emph{(probabilistically)
independent} in $\mu$.

Likewise, we can project a distribution over $\RMem[S]$ to a distribution over
$\RMem[S']$ for $S' \subseteq S$ using the projection $\pi_{S, S'} :
\Dist(\RMem[S]) \to \Dist(\RMem[S'])$, defined as:
\[
  \pi_{S, S'}(\mu)(m_{S'}) \triangleq \sum_{m_S \in \RMem[S] : p_{S'}(m_S) = m_{S'}}
  \mu(m_S) ,
\]
where $p_{S'}$ restricts the range of the memory to $S'$.  For example, $\pi_{S,
S}$ is the identity, while $\pi_{S, \emptyset}$ maps all distributions to
$\delta_0$. We will abbreviate $\pi_{S, S'}$ by $\pi_{S'}$ (or just $\pi$) when
the domains are clear from the context.

\subsection{Probabilistic Programs}

We will work with a variant of the basic probabilistic imperative language
{\PWHILE}. We enforce a clear separation between deterministic and probabilistic
data using simple syntactic conditions, though more sophisticated techniques
(e.g., dataflow analysis) could be also used. Let $\DVar$ be a countable set of
\emph{deterministic variables} disjoint from $\RVar$, and let $\DMem \triangleq
\DVar \to \Val$ be the set of \emph{deterministic memories}, or \emph{stores}.
The expression language is largely standard:
\begin{align*}
  \DExp \ni e_d &::= \DVar \mid \DExp + \DExp \mid \DExp \land \DExp \mid \cdots
  \\
  \RExp \ni e_r &::= \DExp \mid \RVar \mid \RExp + \RExp \mid \RExp \land \RExp \mid \cdots
\end{align*}
We assume that expressions are typed using a simple type system, and we only
work with well-typed expressions. We interpret deterministic expressions as maps
$\denot{e_d} : \DMem \to \Val$. Randomized expressions are interpreted as maps
$\denot{e_r} : \DMem \times \RMem[S] \to \Val$, where $S \subseteq \RVar$
contains all randomized variables in $e_r$; if $e_r$ mentions variables outside
of $S$, then the interpretation is not defined. It is also straightforward to
lift this interpretation to interpret randomized expressions in distributions
over randomized memories: $\denot{e_r} : \DMem \times \Dist(\RMem[S]) \to
\Dist(\Val)$.

Next, we consider the commands. $\RCom$ commands appearing under a randomized
guard---so they cannot assign to deterministic variables---while $\Com$ commands
are general.
\begin{align*}
  \RCom \ni c &::= \Skip
                \mid \Assn{\RVar}{\RExp}
                \mid \Rand{\RVar}{\Unif_S}
                \mid \Seq{\RCom}{\RCom} \\
               &\mid \DCond{\DExp}{\RCom}{\RCom}
                \mid \RCond{\RExp}{\RCom}{\RCom}
                \mid \DWhile{\DExp}{\RCom}
  \\
  \Com \ni c &::= \Skip
                \mid \Assn{\DVar}{\DExp}
                \mid \Assn{\RVar}{\RExp}
                \mid \Rand{\RVar}{\Unif_S}
                \mid \Seq{\Com}{\Com} \\
               &\mid \DCond{\DExp}{\Com}{\Com}
                \mid \RCond{\RExp}{\RCom}{\RCom}
                \mid \DWhile{\DExp}{\Com}
\end{align*}
The main probabilistic command is sampling: $\Rand{\RVar}{\Unif_S}$ takes
a uniform sample from a finite, non-empty set $S$ and assigns it to a variable.

Our grammar separates commands for assignments to deterministic variables and
randomized variables, and separates commands for deterministic and probabilistic
conditionals (we require loop guards to be deterministic). These distinctions
will be important when we introduce our proof system. We will also use a few
standard variants of commands:
\begin{gather*}
  \DFor{i}{1}{N}{c} \triangleq \Seq{\Assn{i}{1}}{\DWhile{i \leq N}{(\Seq{c}{\Assn{i}{i + 1}})}} \\
  \DCondt{b}{c} \triangleq \DCond{b}{c}{\Skip}
  \qquad
  \RCondt{b}{c} \triangleq \RCond{b}{c}{\Skip}
\end{gather*}
We interpret programs using a restricted version of the standard semantics due
to \citet{Kozen81}, assuming additionally that programs terminate on all
inputs---notions like product distribution and probabilistic independence are
poorly behaved when programs may diverge with positive probability. Technically,
programs transform \emph{configurations}, pairs of a deterministic memory
$\sigma$ and a distribution $\mu$ over randomized memories:
\[
  \denot{c} : (\DMem \times \Dist(\RMem)) \to (\DMem \times \Dist(\RMem)) .
\]
\Cref{fig:semantics} presents the program semantics; $\text{Unif}_S \in
\Dist(S)$ is the uniform distribution over a finite, non-empty set $S$, which
assigns probability $1/|S|$ to every element $s \in S$. The semantics of random
conditionals uses convex combination lifted to configurations; this is defined
since the output stores are equal because branches under random guards may not
modify deterministic variables.

\begin{figure}
  \begin{align*}
    \denot{\Skip}(\sigma, \mu)
    &\triangleq (\sigma, \mu) \\
    \denot{\Assn{x_d}{e_d}}(\sigma, \mu)
    &\triangleq (\sigma[x_d \mapsto \denot{e_d}\sigma], \mu) \\
    \denot{\Assn{x_r}{e_r}}(\sigma, \mu)
    &\triangleq (\sigma, \dbind(\mu, m \mapsto \dunit(m[x_r \mapsto \denot{e_r}(\sigma, m)]))) \\
    \denot{\Rand{x_r}{\Unif_S}}(\sigma, \mu)
    &\triangleq (\sigma, \dbind(\mu, m \mapsto \dbind(\text{Unif}_S, u \mapsto \dunit(m[x_r \mapsto u])))) \\
    \denot{\Seq{c}{c'}}(\sigma, \mu)
    &\triangleq \denot{c'} (\denot{c}(\sigma, \mu)) \\
    \denot{\DCond{b}{c}{c'}}(\sigma, \mu)
    &\triangleq \begin{cases}
      \denot{c}(\sigma, \mu) &: \denot{b} \sigma = \true \\
      \denot{c'}(\sigma, \mu) &: \denot{b} \sigma = \false
    \end{cases} \\
    \denot{\RCond{b}{c}{c'}}(\sigma, \mu)
    &\triangleq
    \dconv{\rho}
      {\denot{c}(\sigma, \dcond{\mu}{\denot{b}\sigma = \true})}
      {\denot{c'}(\sigma, \dcond{\mu}{\denot{b}\sigma = \false})} \\
    &\qquad\qquad \text{where } \rho = \mu(\denot{b}\sigma = \true) \\
    \denot{\DWhile{b}{c}}(\sigma, \mu)
    &\triangleq
    \denot{\underbrace{c \mathbin{;} \cdots \mathbin{;} c}_{\mathrlap{N(\sigma) \triangleq \text{ \#steps until } \denot{b} = \false }}}(\sigma, \mu)
  \end{align*}

  \caption{Program semantics}%
  \label{fig:semantics}
\end{figure}

\section{A Probabilistic Model of BI}%
\label{sec:pbi}

Assertions in separation logic are based on the logic of \emph{bunched
implications} (BI)~\citep{DBLP:journals/bsl/OHearnP99,DBLP:conf/lics/Pym99}. We
first review the syntax and semantics of this logic, then introduce a
probabilistic interpretation that will enable BI formulas to describe
probabilistic states.

\subsection{The Syntax and (Resource) Semantics of BI, in Brief}

The logic of bunched implications (BI) is a substructural logic with the
following formulas:\footnote{%
  We do not need the multiplicative identity $\sepid$ from BI, since it will be
equivalent to $\top$ in our setting.}
\[
  \phi, \psi ::= p \in \APred
    \mid \top
    \mid \bot
    \mid \phi \land \psi
    \mid \phi \lor \psi
    \mid \phi \to \psi
    \mid \phi \sepand \psi
    \mid \phi \sepimp \psi
\]
Throughout, $p$ ranges over a set of atomic propositions $\APred$. Negation
$\neg \phi$ is defined as $\phi \to \bot$. Intuitively, BI combines standard
propositional logic with a substructural fragment consisting of the separating
conjunction $\sepand$ and the separating implication (``magic wand'') $\sepimp$.
In the standard heap model of BI underlying separation logic, atomic
propositions describe the contents of particular heap locations, separating
conjunction combines assertions describing disjoint portions of the heap, and
separating implication describes the result of adjoining the current heap with a
disjoint portion.

BI can be given several kinds of semantics.  We follow the \emph{resource}
semantics, as developed by David Pym and
others~\citep{DBLP:journals/tcs/PymOY04}. The basic idea is to define a Kripke
semantics where the set of possible worlds forms a partial, pre-ordered
commutative monoid $\mathcal{M} = (M, \circ, e, \sqsubseteq)$.

\begin{definition}[\citet{DBLP:journals/mscs/GalmicheMJP05}]
  A (partial) \emph{Kripke resource monoid} consists of a set $M$ of possible
  worlds, a partial binary operation $\circ : M \times M \rightharpoonup M$, an
  element $e \in M$, and a pre-order $\sqsubseteq$ on $M$ such that the monoid
  operation
  \begin{itemize}
    \item has identity $e$: for all $x \in M$, we have $e \circ x = x \circ e =
      x$;
    \item is associative: $x \circ (y \circ z) = (x \circ y) \circ x$, where
      both sides are either defined and equal, or both undefined; and
    \item is compatible with the pre-order: if $x \sqsubseteq y$ and $x'
      \sqsubseteq y'$ and if both $x \circ x'$ and $y \circ y'$ are defined,
      then $x \circ x' \sqsubseteq y \circ y'$.
  \end{itemize}
\end{definition}

Under the resource interpretation of BI, possible worlds are collections of
resources, the monoid operation $\circ$ combines sets of resources, and the
identity $e$ represents the lack of resources. The monoid operation may fail to
be defined when combining two incompatible sets of resources; this is useful for
modeling resources that should not be duplicated, such as heap addresses. We
write $(m \circ m')\downarrow$ when the combination is defined.

\begin{definition}
  Let $(M, \circ, e, \sqsubseteq)$ be a partial Kripke resource monoid and let
  $\denot{-} : \APred \to 2^M$ be a \emph{Kripke resource interpretation} of
  atomic formulas: if $m \in \denot{p}$ and $m \sqsubseteq m'$, then $m' \in
  \denot{p}$. Then the corresponding \emph{Kripke resource model} of BI can be
  defined as follows:
  \begin{align*}
    m &\models p &&\text{iff } m \in \denot{p} \\
    m &\models \top &&\text{always} \\
    m &\models \bot &&\text{never} \\
    m &\models \phi \land \psi &&\text{iff }
    m \models \phi \text{ and } m \models \psi \\
    m &\models \phi \lor \psi &&\text{iff }
    m \models \phi \text{ or } m \models \psi \\
    m &\models \phi \to \psi &&\text{iff for all }
    m \sqsubseteq m',\
    m' \models \phi \text{ implies } m' \models \psi \\
    m &\models \phi \sepand \psi &&\text{iff exist }
    m_1, m_2 \text{ with } (m_1 \circ m_2) \downarrow \text{ and } m_1 \circ m_2 \sqsubseteq m
    \text{ such that } m_1 \models \phi \text{ and } m_2 \models \psi \\
    m &\models \phi \sepimp \psi &&\text{iff for all } m'
    \text{ such that } m' \models \phi,\
    (m \circ m') \downarrow \text{ implies } m \circ m' \models \psi
  \end{align*}
  All formulas satisfy the \emph{Kripke monotonicity} property: if $m \models
  \phi$ and $m \sqsubseteq m'$, then $m' \models \phi$ as well. We write
  $\models \phi$ when $\phi$ is \emph{valid}, i.e., when $\phi$ holds in all
  worlds.
\end{definition}

BI enjoys good metatheoretic properties and a rich proof theory. Many models are
known beyond heaps, including doubly closed categories (DCCs), presheafs, and
Petri nets. There are also complete proof systems for {BI}. The interested
reader should consult \citet{PymMono} or \citet{DochertyThesis} for a detailed
treatment of BI's proof theory, and \citet{DBLP:journals/mscs/GalmicheMJP05} for
more information about the partial monoid semantics we use here.

\subsection{A Probabilistic Version of BI}

By leveraging the resource semantics of BI, we can give a probabilistic
interpretation of BI formulas.

\begin{definition}
  Let $M$ be the set of program configurations $\DMem[S] \times \Dist(\RMem[T])$
  where $S$ ranges over subsets of $\DVar$ and $T$ ranges over subsets
  of $\RVar$. Let $\circ$ be a partial binary operation defined as:
  \[
    (\sigma, \mu) \circ (\sigma', \mu') \triangleq \begin{cases}
      (\sigma \cup \sigma', \mu \otimes \mu')
      &: \sigma = \sigma' \text{ on } \dom(\sigma) \cap \dom(\sigma')
      \text{ and } \dom(\mu) \cap \dom(\mu') = \emptyset \\
      \text{undefined} &: \text{otherwise} .
    \end{cases}
  \]
  Let $e$ be the empty deterministic memory paired with the Dirac distribution
  over the empty probabilistic memory, and let $\sqsubseteq$ be the following
  partial order:
  \[
    (\sigma, \mu) \sqsubseteq (\sigma', \mu')
    \quad\text{iff}\quad \begin{cases}
    &\dom(\sigma) \subseteq \dom(\sigma') \text{ and }
    \sigma = \sigma' \text{ on } \dom(\sigma) \\
    &\dom(\mu) \subseteq \dom(\mu') \text{ and }
    \mu = \pi_{\dom(\mu'), \dom(\mu)}(\mu') .
    \end{cases}
  \]
  Then $(M, \circ, e, \sqsubseteq)$ is a Kripke resource monoid.
\end{definition}

To describe basic properties of configurations, we take the following atomic
formulas.
\begin{align*}
  \APred \ni p &::= \Unif_S[\RExp] \mid \RExp \sim \RExp \mid \DExp = \DExp \mid \DExp \leq \DExp \mid \cdots
\end{align*}
We will fix a base theory $E$ of program expressions, enough to interpret the
necessary arithmetic operations $(+, \cdot)$ and relations $(=, \leq)$, and we
write $\models_E \phi$ if $\phi$ when $\phi$ is valid. For example $\models_E e
+ e' = e' + e$ holds for any two expressions, randomized or not. 

Validity for atomic formulas of deterministic expressions is defined as
expected: these formulas can be interpreted as subsets of $\DMem$. More
formally, for any deterministic proposition $p_d$ we write $\sigma \models_E
p_d$ if $p_d$ holds in $\sigma$, and we define:
\[
  (\sigma, \mu) \models p_d \text{ iff } \sigma \models_E p_d
\]
The more interesting cases are the atomic formulas for randomized expressions.

\begin{definition}
  For a nonempty finite set $S$ and a randomized expression $e_r \in \RExp$, we
  define $\denot{\Unif_S[e_r]}$ to be the set of configurations $(\sigma, \mu)$
  where $\FV(e_r) \subseteq \dom(\sigma) \cup \dom(\mu)$, and
  $\denot{e_r}(\sigma, \mu)$ assigns probability $1/|S|$ to each element of $S$;
  we omit $S$ when it is clear from the context.

  We define $\denot{e_r \sim e_r'}$ to be the set of configurations $(\sigma,
  \mu)$ where $\FV(e_r) \cup \FV(e_r') \subseteq \dom(\sigma) \cup \dom(\mu)$ and
  $\denot{e_r}(\sigma, m) = \denot{e_r'}(\sigma, m)$ for $m \in \supp(\mu)$.
  This formula asserts equality between randomized expressions; we use $\sim$ to
  avoid confusion with equality between deterministic expressions. We abbreviate
  $\Dist[e] \triangleq e \sim e$. Explicitly, $\denot{\Dist[e]}$ contains all
  configurations $(\sigma, \mu)$ where $\FV(e) \subseteq \dom(\sigma) \cup
  \dom(\mu)$.
\end{definition}

Since the interpretation of atomic assertions is monotonic, our configurations
are a Kripke resource model of BI. An important feature of the resulting
semantics is that validity only depends on the variables in the formula. (We
defer proofs to \fullref{app:proofs}.)

\begin{restatable}[Restriction]{lemma}{LEMrestriction}\label{lem:bi-restriction}
  Let $(\sigma, \mu)$ be any configuration and let $\phi$ be a BI formula. Then:
  \[
    (\sigma, \mu) \models \phi \iff (\sigma, \pi_{\FV(\phi)}(\mu)) \models \phi .
  \]
\end{restatable}

One useful consequence is the following property, which allows a $\land$
conjunct to be pulled into a $\sepand$ conjunct covering all of the formula's
free random variables.

\begin{restatable}[Extrusion]{lemma}{LEMextrusion}\label{lem:bi-exchange}
  If $\models \phi \to \Dist[FV(\eta) \cap \RVar]$, then $\models (\phi \sepand
  \psi) \land \eta \to (\phi \land \eta) \sepand \psi$.
\end{restatable}

An easy and useful consequence follows when $\eta$ does not mention any random
variables.\footnote{%
  For readers familiar with separation logic, deterministic propositions
resemble \emph{pure} assertions in the heap model of BI.}

\begin{corollary}
  Let $p_d$ be a deterministic proposition. The following axiom is sound:
  \[
    \models \psi \land p_d \to \psi \sepand p_d
  \]
\end{corollary}
\begin{proof}
  Since $\top$ is the unit for $\sepand$ in our semantics, $\psi \land p_d$
  implies $(\top \sepand \psi) \land p_d$. Since $p_d$ does not mention any
  random variables, \cref{lem:bi-exchange} implies $(\top \land p_d) \sepand
  \psi$. Symmetry of $\sepand$ gives $\psi \sepand p_d$.
\end{proof}

We briefly mention two other important features of our semantics. First, the
semantics is intuitionistic: $\phi \lor \neg \phi$ is not valid. Second, our
semantics admits weakening:
\[
  \models \phi \sepand \psi \to \phi \land \psi
\]
We will use repeatedly this property to pull out facts about specific variables
from a larger assertion.

\subsection{Axiom Schema for Atomic Formulas}
Next, we present our axioms for atomic formulas. Much like the situation for
atomic formulas in the ``pointer logic'' underlying standard separation logic,
these axioms are not complete. Nevertheless, they are already sufficient to
reason about many interesting probabilistic programs. 

We begin with axioms for formulas with $\sim$. The main difference between
$\sim$ and standard equality is that since $\sim$ is interpreted in a randomized
configuration---which might not have all of $\RVar$ in its domain---replacing
expressions by equal expressions must not introduce new random variables.

\begin{restatable}{lemma}{LEMsim}
  The following axiom schema are valid:
  \begin{align}
    &\models e_r \sim e_r' \to e_r' \sim e_r
    \tag{S1}
    \\
    &\models e_r \sim e_r' \land e_r' \sim e_r'' \to e_r \sim e_r''
    \tag{S2}
    \\
    &\models e_r \sim e_r' \to e_r \sim e_r'' \text{ whenever }
    \models_E e_r' = e_r'' \text{ and } FV(e_r'') \cap \RVar \subseteq FV(e_r')
    \cap \RVar
    \tag{S3}
    \\
    &\models e_r \sim e_r \to e_r' \sim e_r'
    \text{ whenever } FV(e_r') \cap \RVar \subseteq FV(e_r) \cap \RVar
    \tag{S4}
  \end{align}
\end{restatable}

Note that $\models e_r \sim e_r$ is \emph{not} an axiom---it is not sound, since
it may not hold in a randomized memory $\Dist(\RMem[\emptyset])$ with empty
domain. We also have axioms for uniformity propositions.

\begin{restatable}{lemma}{LEMunif}
  The following axiom schema are valid:
  \begin{align}
    &\models e_r \sim e_r' \land \Unif_{S}[e_r] \to \Unif_{S}[e_r']
    \tag{U1}
    \\
    &\models \Unif_{S}[e_r] \to e_r \sim e_r
    \tag{U2}
    \\
    &\models \Unif_{S}[e_r] \to \Unif_S[f(e_r)] \text{ for any bijection } \denot{f} : S \to S
    \text{ and } \FV(f) \cap \RVar \subseteq \FV(e_r) \cap \RVar
    \tag{U3}
  \end{align}
\end{restatable}

\subsection{Comparison with Typical Models of BI}

This subsection contains a more detailed comparison with other models of {BI};
readers who are primarily interested in the separation logic can safely skip
ahead to \cref{sec:psl}.

Our model of BI is strongly inspired by the standard heap model. There, worlds
are partial maps from heap locations to values and the main atomic assertion $e
\mapsto e'$ (``points-to'') indicates that in the current heap, the location
stored in expression $e$ holds the value denoted by $e'$ in the current store. A
separating conjunction of two points-to assertions $e \mapsto e' \sepand f
\mapsto f'$ indicates that the addresses held in $e$ and $f$ do not alias. This
separation property cannot be deduced syntactically---two expressions mentioning
different variables may refer to the same heap location in the current
store---but it is crucial for local reasoning in the presence of aliasing.

Our probabilistic model is designed to capture a fundamentally different notion
of separation that is natural to the probabilistic setting. The randomized
portion of the program state plays the role of the heap in the heap model, but
randomized variables are not heap-allocated. Accordingly, the names of
randomized variables are fixed and it is possible to syntactically determine
when two assertions refer to the same variable. However, it is \emph{not} always
possible to determine when two assertions refer to \emph{probabilistically
independent} variables---for instance, the assertion $\Unif[x] \land \Unif[y]$
holds in any memory where $x$ and $y$ are uniform, but $x$ and $y$ may be
correlated. The stronger property of probabilistic independence---the notion of
non-aliasing that probabilistic BI is designed to model---is captured by the
assertion $\Unif[x] \sepand \Unif[y]$.

As we have noted, our semantics is intuitionistic. This aspect stems from our
choice of a non-trivial partial order $\sqsubseteq$ over worlds. If this order
is taken to be discrete, relating only identical elements, and atomic formulas
are interpreted exactly, referring to the entire memory rather than a fragment,
we would arrive at a classical or \emph{Boolean} semantics for probabilistic {BI}.
This situation is mirrored in heap models of BI, where the classical logic BBI
has a semantics with a discrete order.

For heap models, BBI is more precise than BI---it supports atomic assertions
that are not preserved under heap extension, like $\text{emp}$ for empty heap,
and BI formulas $\phi$ can be recovered by BBI formulas $\phi \sepand \top$.  In
the probabilistic setting, however, a classical logic seems to run into trouble.
For instance, it is \emph{not} the case that a formula $\phi$ valid in $(\sigma,
\mu)$ under our semantics can be directly translated into a classical semantics:
$(\sigma, \mu) \models \phi \sepand \top$ would state that the domain of the
distribution modeling $\phi$ is probabilistically independent of all other
variables, which does not follow from $(\sigma, \mu) \models \phi$. We leave the
development of a classical version of probabilistic BI for future work.
 
\section{A Probabilistic Separation Logic}%
\label{sec:psl}

We now have all the ingredients needed for our separation logic {\SYSTEM}.
First, the judgments.

\subsection{Judgments and Validity}

\begin{definition}
  \SYSTEM{} judgments have the form $\psl{\phi}{c}{\psi}$ where $\phi$ and
  $\psi$ are probabilistic BI formulas. Such a judgment is \emph{valid}, denoted
  $\models \psl{\phi}{c}{\psi}$, if for all configurations $(\sigma, \mu) \in
  \DMem[\DVar] \times \Dist(\RMem[\RVar])$ satisfying $(\sigma, \mu) \models
  \phi$, we have $\denot{c}(\sigma, \mu) \models \psi$.
\end{definition}

We have defined validity to quantify over only input states with all variables
in the domain. Unlike in separation logic, programs do not allocate variables
and memory faults are not possible, so there is no reason to consider behaviors
from partial configurations in the program logic.

\subsection{Proof Rules: Deterministic Constructs}

\begin{figure}
  \begin{mathpar}
    \hspace*{1cm}\inferrule*[Left=DAssn]
    {~}
    { \vdash \psl{\psi[e_d/x_d]}{\Assn{x_d}{e_d}}{\psi} }
    \and
    \inferrule*[Left=Skip]
    {~}
    { \vdash \psl{\phi}{\Skip}{\phi} }
    \and
    \inferrule*[Left=Seqn]
    { \vdash \psl{\phi}{c}{\psi} \\ \vdash \psl{\psi}{c'}{\eta} }
    { \vdash \psl{\phi}{\Seq{c}{c'}}{\eta} }
    \\
    \inferrule*[Left=DCond]
    { \vdash \psl{\phi \land b = \ktt}{c}{\psi}
    \\\\ \vdash \psl{\phi \land b = \kff }{c'}{\psi} }
    { \vdash \psl{\phi}{\DCond{b}{c}{c'}}{\psi} }
    \and
    \inferrule*[Left=DLoop]
    { \vdash \psl{\phi \land b = \ktt}{c}{\phi} }
    { \vdash \psl{\phi}{\DWhile{b}{c}}{\phi \land b = \kff} }
  \end{mathpar}
  \caption{Proof rules: deterministic constructs}%
  \label{fig:rules-det}
\end{figure}

We introduce the proof system of \SYSTEM{} in three stages. First, we consider
the deterministic constructs in \cref{fig:rules-det}. The rule \rname{DAssn} is
the usual Hoare rule for assignments, but it is only sound for assignments to
\emph{deterministic} variables. Otherwise, the rules are as expected.

When proving judgments of for-loops, we will use the following derived rule:
\[
  \inferrule*[Left=DFor]
  { \vdash \psl{\phi}{c}{\phi[i + 1/i]} \\ \FV(N) \cap \MV(c) = \emptyset }
  { \vdash \psl{\phi[1/i]}{\DFor{i}{1}{N}{c}}{\phi[N+1/i]} }
\]
$\MV(c)$ is the set of variables that may be modified by $c$; we defer the
formal definition to \cref{def:var-cond}, when we discuss the frame rule.

\subsection{Proof Rules: Probabilistic Constructs}

\begin{figure}
  \begin{mathpar}
    \inferrule*[Left=RAssn]
    { x_r \notin \FV(e_r) }
    { \vdash \psl{\top}{\Assn{x_r}{e_r}}{x_r \sim e_r} }
    \and
    \inferrule*[Left=RSamp]
    {~}
    { \vdash \psl{\top}{\Rand{x_r}{\Unif_S}}{\Unif_S[x_r]} }
    \\
    \inferrule*[Left=RDCond]
    { \vdash \psl{\phi \land b \sim \ktt }{c}{\psi}
      \\ \vdash \psl{\phi \land b \sim \kff }{c'}{\psi}
    \\ \models \phi \to (b \sim \ktt \lor b \sim \kff) }
    { \vdash \psl{\phi}{\RCond{b}{c}{c'}}{\psi} }
    \\
    \inferrule*[Left=RCond]
    { \vdash \psl{\phi \sepand b \sim \ktt}{c}{\psi \sepand b \sim \ktt}
      \\ \vdash \psl{\phi \sepand b \sim \kff}{c'}{\psi \sepand b \sim \kff}
    \\ \psi \in \SP }
    { \vdash \psl{\phi \sepand \Dist[b]}{\RCond{b}{c}{c'}}{\psi \sepand \Dist[b] } }
  \end{mathpar}
  \caption{Proof rules: probabilistic constructs}%
  \label{fig:rules-prob}
\end{figure}

\cref{fig:rules-prob} presents the proof rules for randomized operations.
\rname{RAssn} and \rname{RSamp} are for randomized assignment and random
sampling, respectively; in contrast to \rname{DAssn}, these rules reason
forwards. Both rules are presented in their ``local'' form, where the
pre-condition is trivial. We will soon derive ``global'' variants, with general
pre-conditions, from the structural rules.

There are two rules for randomized conditionals. \rname{RDCond} resembles
\rname{DCond}, with a side-condition to ensure that the randomized guard $b$ is
deterministic. \rname{RCond} applies when the guard is truly probabilistic, and
it has two unusual aspects. First, the pre-condition in the conclusion requires
the guard to be separated from the rest of the pre-condition; that is, the guard
must be probabilistically independent of the portion of the randomized memory
satisfying $\phi$. This separation is crucial for $\phi$ to be soundly used as a
pre-condition in each branch: the input distribution to each branch is obtained
by \emph{conditioning} on the value of the guard expression in the input
distribution. This operation may not preserve $\phi$, even if $\phi$ and the
guard have no variables in common---this is a particular feature of the
probabilistic setting.

\begin{example}
  Suppose that $x, y, z$ are randomized boolean variables, and let $\mu$ be the
  output of:
  \[
    \Rand{x}{\Unif_\mathbb{B}}; \Rand{y}{\Unif_\mathbb{B}}; \Assn{z}{x \lor y}
  \]
  In words, $x$ and $y$ store the results of two fair coin flips, and $z$ stores
  the value of $x \lor y$. Then $x$ and $y$ are independent in $\mu$, i.e.,
  $\Dist[x] \sepand \Dist[y]$ holds in $\mu$. However, if $S \subseteq
  \RMem[\RVar]$ is the set of all randomized memories where $z = \ktt$,
  representing the event that $z$ is true, then $\Dist[x] \sepand \Dist[y]$ does
  not hold in $\dcond{\mu}{S}$. Intuitively, if we know $z = \ktt$, then $x$ and
  $y$ are correlated: if one is false, then the other must be true.
\end{example}

\rname{RCond} also shows that the guard remains independent of the branch
post-condition assuming the branches do not modify the guard, and the branch
post-condition determines a \emph{unique} portion of the distribution over
randomized memories. Formally, we adapt the following class of assertions from
separation logic~\citep{reynolds-SL}.

\begin{definition}\label{def:SP}
  A formula $\phi$ is \emph{supported} ($\SP$) if for any deterministic memory
  $\sigma$, there exists a randomized memory $\mu$ such that if $(\sigma, \mu')
  \models \phi$, then $\mu \sqsubseteq \mu'$.
\end{definition}

The following syntactic conditions ensure $\SP$.

\begin{lemma}\label{lem:SP}
  The following assertions are $\SP$:
  \begin{align*}
    \eta ::= p_d
    \mid \RVar \sim \DExp
    \mid \Unif_S[\RVar]
    \mid \eta \sepand \eta
  \end{align*}
\end{lemma}
\begin{proof}
  By induction on $\eta$. The base cases are immediate: $p_d$ holds in the
  unique randomized memory with empty domain, while $x \sim e$ and $\Unif_S[x]$
  hold in unique randomized memories with domain $\{ x \}$. The inductive case
  is also straightforward.
\end{proof}

\begin{example}[Non-$\SP$ assertions]
  A simple example of an assertion that is not covered by \cref{lem:SP} is $\phi
  \triangleq \Dist[x]$, where $x \in \RVar$ is a boolean randomized variable. It
  is easy to see that $\phi$ is not $\SP$; for instance, $\phi$ holds in two
  incomparable distributions $\delta_{(x \mapsto \ktt)}$ and $\delta_{(x \mapsto
  \kff)}$. Indeed, allowing $\phi$ as a branch post-condition in \rname{RCond}
  would be unsound. Consider the following program:
  \[
    c \triangleq \RCond{b}{\Assn{x}{\ktt}}{\Assn{x}{\kff}}
  \]
  Clearly, $\phi \sepand b \sim \ktt$ and $\phi \sepand b \sim \kff$ are sound
  post-conditions for the two branches. But $c$ is semantically equal to
  $\Assn{x}{b}$, and $\Dist[x] \sepand \Dist[b]$ is not a sound post-condition.
\end{example}

In \fullref{app:CM}, we consider a variant of \rname{RCond} that proves a weaker
post-condition, but relaxes the requirement on $\psi$ and allows the branches to
modify the guard.

\subsection{Structural Rules}

\begin{figure}
  \begin{mathpar}
    \inferrule*[Left=Weak]
    { \vdash \psl{\phi}{c}{\psi} \\ \models \phi' \to \phi \land \psi \to \psi' }
    { \vdash \psl{\phi'}{c}{\psi'} }
    \and
    \inferrule*[Left=True]
    {~}
    { \vdash \psl{\top}{c}{\top} }
    \\
    \inferrule*[Left=Conj]
    { \vdash \psl{\phi_1}{c}{\psi_1} \\
    \vdash \psl{\phi_2}{c}{\psi_2} }
    { \vdash \psl{\phi_1 \land \phi_2}{c}{\psi_1 \land \psi_2} }
    \and
    \inferrule*[Left=Case]
    { \vdash \psl{\phi_1}{c}{\psi_1} \\
    \vdash \psl{\phi_2}{c}{\psi_2} }
    { \vdash \psl{\phi_1 \lor \phi_2}{c}{\psi_1 \lor \psi_2} }
    \\
    \inferrule*[Left=RCase]
    { \vdash \psl{\phi \sepand b \sim \ktt}{c}{\psi \sepand b \sim \ktt}
      \\ \vdash \psl{\phi \sepand b \sim \kff}{c}{\psi \sepand b \sim \kff}
    \\ \psi \in \SP }
    { \vdash \psl{\phi \sepand \Dist[b]}{c}{\psi \sepand \Dist[b]} }
    \and
    \inferrule*[Left=Const]
    { \vdash \psl{\phi}{c}{\psi} \\ \FV(\eta) \cap \MV(c) = \emptyset }
    { \vdash \psl{\phi \land \eta}{c}{\psi \land \eta} }
    \\
    \hspace*{1cm}\inferrule*[Left=Frame]
    { \vdash \psl{\phi}{c}{\psi} \\
      \FV(\eta) \cap \MV(c) = \emptyset \\
      \FV(\psi) \subseteq T \cup \RV(c) \cup \WV(c) \\
    \models \phi \to \Dist[T \cup \RV(c)]}
    { \vdash \psl{\phi \sepand \eta}{c}{\psi \sepand \eta} }
  \end{mathpar}
\caption{Structural rules}%
\label{fig:rules-struct}
\end{figure}

\cref{fig:rules-struct} collects the final group of rules, the structural rules.
\rname{Weak}, \rname{True}, \rname{Conj}, and \rname{Case} are standard;
\rname{RCase} is an analog of \rname{RCond}. The last two rules are more
interesting. \rname{Const} is the rule of constancy from Hoare logic, which
states that formulas $\eta$ that do not mention any of $c$'s modified variables
$\MV(c)$ may be conjoined to the pre- and post-condition. This rule is
\emph{not} sound in standard separation logic---motivating the separating
conjunction and the frame rule---but it \emph{is} sound in {\SYSTEM}: writes
cannot invalidate assertions about other variables.

But, the post-condition in \rname{Const} does not ensure that $\psi$ and $\eta$
refer to probabilistically independent variables. For this stronger guarantee,
we need \rname{Frame}. The side conditions mention several classes of variables.
Roughly speaking, $\RV(c)$ is the set of variables that $c$ may read from, while
$\WV(c)$ is the set of variables that $c$ \emph{must} write to (before possibly
reading from). $\MV(c)$ is the set of variables that $c$ \emph{may} write to, so
$\WV(c)$ is a subset of $\MV(c)$.

We can approximate these sets using a simple syntactic condition.

\begin{definition}\label{def:var-cond}
  $\RV, \WV, \MV$ are defined as follows:
  \begin{mathpar}
    \RV(\Assn{x_r}{e_r}) \triangleq \FV(e_r) \and
    \RV(\Rand{x_r}{\Unif_S}) \triangleq \emptyset \and
    \RV(\Seq{c}{c'}) \triangleq \RV(c) \cup (\RV(c') \setminus \WV(c)) \and
    \RV(\RCond{b}{c}{c'}) \triangleq \FV(b) \cup \RV(c) \cup \RV(c') \and
    \RV(\DWhile{b}{c}) \triangleq \RV(c)
  \end{mathpar}
  \hrule
  \begin{mathpar}
    \WV(\Assn{x_r}{e_r}) \triangleq \{ x_r \} \setminus \FV(e_r) \and
    \WV(\Rand{x_r}{\Unif_S}) \triangleq \{ x_r \}  \and
    \WV(\Seq{c}{c'}) \triangleq \WV(c) \cup (\WV(c') \setminus \RV(c)) \and
    \WV(\RCond{b}{c}{c'}) \triangleq (\WV(c) \cap \WV(c')) \setminus \FV(b)
  \end{mathpar}
  \hrule
  \begin{mathpar}
    \MV(\Assn{x_r}{e}) \triangleq \{ x_r \} \and
    \MV(\Rand{x_r}{\Unif_S}) \triangleq \{ x_r \} \and
    \MV(\Seq{c}{c'}) \triangleq \MV(c) \cup \MV(c') \and
    \MV(\RCond{b}{c}{c'}) \triangleq \MV(c) \cup \MV(c') \and
    \MV(\DWhile{b}{c}) \triangleq \MV(c)
  \end{mathpar}
\end{definition}

Other analyses are possible, so long as non-modified variables are preserved
from input to output, and output modified variables depend only on input read
variables.

\begin{lemma}[Soundness for \RV, \WV, \MV]\label{lem:fv}
  Let $(\sigma', \mu') = \denot{c}(\sigma, \mu)$, and let $S_r = \RV(c), S_w =
  \WV(c), S_c = \RVar \setminus \MV(c)$. Then:
  \begin{enumerate}
    \item Variables outside of $\MV(c)$ are not modified: $\pi_{S_c}(\mu') =
      \pi_{S_c}(\mu)$.
    \item The sets $S_r$ and $S_w$ are disjoint.
    \item There exists $F : \RMem[S_r] \to \Dist(\RMem[\MV(c)])$ such that $\mu'
      = \dbind(\mu, m \mapsto F(\pi_{S_r}(m)) \otimes \dunit(\pi_{S_c}(m)))$.
  \end{enumerate}
\end{lemma}

Returning to \rname{Frame}, we consider the side-conditions one by one.  The
first side-condition is as in \rname{Const}; the framing condition $\eta$ cannot
mention any possibly-modified variables. The second condition states that the
post-condition $\psi$ can only mention variables that are (i) in the footprint
of $\phi$, or (ii) written by $c$. The last condition states that any portion of
the randomized memory satisfying $\phi$ must have a footprint containing $T$ and
all variables read by $c$. Intuitively, these side-conditions ensure that if
the framing condition (i) does not mention modified variables and (ii) is
initially independent of all read variables, then it is independent of all
variables in $\phi$ as well as all written variables---these variables can only
depend on read variables, which were initially all independent from the framing
condition.

\begin{example}
  Using \rname{Const}, we can derive the following global version of
  \rname{RAssn}:
  \[
    \inferrule*[Left=RAssn*]
    { x_r \notin \FV(\phi, e_r) }
    { \vdash \psl{\phi}{\Assn{x_r}{e_r}}{\phi \land x_r \sim e_r} }
  \]
  The set of modified variables is $\{ x_r \}$.
\end{example}

\begin{example}
  Using \rname{Frame}, we can derive the following global version of
  \rname{RSamp}:
  \[
    \inferrule*[Left=RSamp*]
    { x_r \notin \FV(\phi) }
    { \vdash \psl{\phi}{\Rand{x_r}{\Unif_S}}{\phi \sepand \Unif_S[x_r]} }
  \]
  There are no read variables, and the modified and written variables are both
  $\{ x_r \}$.
\end{example}

\subsection{Soundness}

As expected, the proof system is sound.

\begin{restatable}[Soundness]{theorem}{THMsoundness}\label{thm:soundness}
  If $\vdash \psl{\phi}{c}{\psi}$ is derivable, then it is valid: $\models
  \psl{\phi}{c}{\psi}$.
\end{restatable}

We discuss other meta-theoretical properties in \cref{sec:conclusion}.

\section{Examples: Cryptographic Security}%
\label{sec:examples}

We demonstrate our logic by proving security for several cryptographic schemes.
As we will see, our logic can express and prove two distinct forms of
information-theoretic security properties. For more convenient encoding of the
protocols, we will work with an extended language with arrays, which can be
indexed or assigned to via $x[i]$, and finite tuples, which can be indexed or
assigned to via $x.1, x.2$, etc. To write compact assertions about arrays, our
assertions will use big versions of the conjunctions, written $\bigwedge_{i \in
p(i)} \phi(i)$ and $\bigsep_{i \in p(i)} \phi(i)$ where $i$ is a fresh logical
variable and $p$ is deterministic and holds for at most finitely many indices.
In some examples we use assignments of the form $\Assn{x}{e}$ where $x \in
\FV(e)$. Since our assignment rule \rname{Rassn} does not apply here, these
assignments are short for $\Assn{x_f}{e}; \Assn{x}{x_f}$, where $x_f$ is a fresh
temporary variable. We give proof sketches in this section; details can be found
in \fullref{app:examples}.

We will need axioms relating uniformity assertions, pairing, and modular
arithmetic, the main arithmetic operation in our examples. In general, axioms
are strongly dependent on the equational theory of expressions, and it is not
clear how to give a complete axiomatization even for just the modular addition
operator; we give axioms schema that are broadly useful for our examples.

\begin{restatable}{lemma}{LEMaxioms}\label{lem:axioms}
  Let $q \geq 2$ be any integer, and let $\{ x_i \}$ be any finite set of
  distinct variables. The following axiom schema are sound.
  \begin{align}
    &\models \Unif_{S_1}[x_1] \sepand \cdots \sepand \Unif_{S_n}[x_n]
      \leftrightarrow \Unif_{S_1 \times \cdots \times S_n}[(x_1, \dots, x_n)]
      \notag
      \\
    &\models \Unif_{\mathbb{Z}_q} [x_1] \sepand \Dist[x_2] \sepand \cdots
    \sepand \Dist[x_n] \land x_0 \sim x_1 + \cdots + x_n \text{ mod } q \to
      \Unif_{\mathbb{Z}_q}[x_0] \sepand \Dist[x_2] \sepand \cdots \sepand \Dist[x_n]
      \notag
    \\
    \intertext{In particular, we will use two derived axioms (writing $\oplus$ for xor, addition modulo $2$):}
    &\models \Unif_{\mathbb{Z}_2}[x_1] \sepand \Dist[x_2] \land x_0 \sim x_1
    \oplus x_2 \to \Unif_{\mathbb{Z}_2}[x_0] \sepand \Dist[x_2]
    \label{ax:xor} \tag{U4} \\
    &\models \Unif_{\mathbb{Z}_q}[x_1] \sepand \Dist[x_2] \sepand \Dist[x_3] \land x_0 \sim x_1 + x_2 +
    x_3 \text{ mod } q \to \Unif_{\mathbb{Z}_q}[x_0] \sepand \Dist[x_2] \sepand \Dist[x_3]
    \label{ax:mod} \tag{U5}
  \end{align}
  These axioms also hold for expressions with at most one free variables.
\end{restatable}

\subsection{Warming Up: The One-Time Pad}

The one-time pad (OTP) is a simple encryption scheme~\cite{katz2014introduction}
enjoying a strong property called \emph{perfect secrecy}. The OTP is a triple of
algorithms $(\mathsf{Gen}, \mathsf{Enc}, \mathsf{Dec})$ parameterized by
$n \in \mathbb{N}$, which determines the key space $\mathcal{K}$, message space
$\mathcal{M}$, and ciphertext space $\mathcal{C}$, each equal to ${\{0,1\}}^n$:
\begin{itemize}
  \item $\mathsf{Gen}$: Select a key $k \in \mathcal{K}$ uniformly at random.
  \item $\mathsf{Enc}_k(m)$: Given a key $k \in \mathcal{K}$ and message
  $m \in \mathcal{M}$, output the ciphertext $c = m \oplus k$.
  \item $\mathsf{Dec}_k(c)$: Given a key $k \in \mathcal{K}$ and ciphertext
  $c \in \mathcal{C}$, output the message $m = c \oplus k$.
\end{itemize}
We model \textsf{Gen} and \textsf{Enc} with the following code, where $m$ is a
deterministic input variable:
\[
\begin{array}{l}
  \Rand{k}{{\{0,1\}}^n};
  \Assn{c}{m \oplus k}
\end{array}
\]
But why is the OTP perfectly secret, and what does that mean in the first place?
A natural way to define secrecy is to say that ``the ciphertext reveals nothing
about the plaintext.'' We can formalize this intuitive notion in two ways. One
variant requires that an observer's view is uniformly distributed for all
private inputs.
\begin{definition}[Perfect Secrecy as Probabilistic Non-Interference]\label{def:ps2}
  An encryption scheme $\Pi : \mathcal{M} \to \Dist(\mathcal{C})$ is perfectly
  secret if for every pair of messages $m,m' \in \mathcal{M}$ and every
  ciphertext $c \in \mathcal{C}$, we have:
  \[
    \Pr[\Pi(m) = c] = \Pr[\Pi(m')=c].
  \]
\end{definition}
\noindent This notion can be seen as \emph{probabilistic non-interference}, a
generalization of a standard information flow property to the probabilistic
setting. As previously noted, perfect secrecy follows from uniformity: if we can
show that the output distribution is uniformly distributed, an observer's view
is the same for all private inputs since the uniform distribution is unique. In
other words, it suffices to show $\Unif[c]$ as a post-condition, where $c$
denotes the (probabilistic) output of $\mathsf{Enc}_k(m)$.

Another way to define security is to treat the secret input---here, the
message---as drawn from a distribution, and then require the observer's view to
be probabilistically independent of the input for every distribution on inputs.
This formulation captures security through an intuitive reading of independence:
the public output reveals no new information about the secret input.

\begin{definition}[Perfect Secrecy as Input Independence]\label{def:ps3}
  Let $\mathcal{M}$ be the message space and let $\mathcal{C}$ be the ciphertext
  space. Regarding an encryption scheme $\Pi : \mathcal{M} \to
  \Dist(\mathcal{C})$ as the map $\tilde{\Pi} : \Dist(\mathcal{M}) \to
  \Dist(\mathcal{M} \times \mathcal{C})$ that preserves the input distribution,
  $\Pi$ is perfectly secret if  the random variables $M$ and $C$ are independent
  in $\tilde{\Pi}(\mu)$ for every input distribution $\mu$ over $\mathcal{M}$.
\end{definition}

\noindent Again, we are done if we can show that the output distribution is
independent from the input distribution.  In other words, it suffices to show
$\Dist[m] \sepand \Dist[c]$ as a post-condition, where $m$ denotes the secret
input and $c$ denotes the (probabilistic) output of $\mathsf{Enc}_k(m)$.

Although it turns out that these two formulations of perfect secrecy are
equivalent~\cite{katz2014introduction}, the proofs of these properties differ.
To demonstrate, we show that the OTP is perfectly secret according to both
definitions by establishing $\Unif[c]$ and $\Dist[m] \sepand \Dist[c]$ as
post-conditions.

\subsubsection{Proof of Uniformity}
Starting from the trivial pre-condition $\Phi_1 \triangleq \top$, we would like
to prove the post-condition $\Psi \triangleq \Unif[c]$. Perfect secrecy then
follows from \cref{def:ps2}. We start by adjoining the random sample for $k$
according to \rname{RSamp}:
\[
  \Unif[k].
\]
Then, by \rname{RAssn*}, assigning to $c$ gives
\[
  \Unif[k] \mathrel{\land} c \sim m \mathrel{\oplus} k.
\]
The xor axiom (\ref{ax:xor}) gives the desired post-condition:
\[
  \Psi \triangleq \Unif[c].
\]

\subsubsection{Proof of Input Independence}
Starting from the pre-condition $\Phi_1 \triangleq \Dist[m]$, we would like to
prove the post-condition $\Psi \triangleq \Dist[m] \sepand \Unif[c]$, now
treating $m$ as a randomized variable; perfect secrecy then follows from
\cref{def:ps3}. We start by using \rname{RSamp*} to adjoin the sample $k$:
\[
  \Dist[m] \sepand \Unif[k].
\]
By \rname{RAssn*}, assigning to $c$ then gives
\[
  \Dist[m] \sepand \Unif[k] \mathrel{\land} c \sim m \mathrel{\oplus} k.
\]
Finally, applying the xor axiom (\ref{ax:xor}) gives
\[
  \Dist[m] \sepand \Unif[c],
\]
which implies the desired post-condition:
\[
  \Psi \triangleq \Dist[m] \sepand \Dist[c].
\]

We now prove both properties for several other constructions.

\subsection{Private Information Retrieval}

Private information retrieval (PIR) enables a user to retrieve an item from a
database server without the server learning which item was
requested~\citep{chor1995private}. For instance, one (highly inefficient) scheme
just has the server send the entire database to the user. When multiple copies
of the database are held by multiple servers, however, significantly more
efficient PIR schemes are possible.

We consider the following single-bit two-server PIR scheme by
\citet{chor1995private}. Two \emph{non-colluding} servers
$\mathcal{S}_0,\mathcal{S}_1$ store the same $N$-bit database $D$. A client
$\mathcal{C}$ wishes to access the $i$-th bit of $D$, denoted $D[i]$, without
either server learning any information about $i$. To achieve this, $\mathcal{C}$
first uniformly samples an $N$-bit string $q_0$. Then, $\mathcal{C}$ sends $q_0$
to $\mathcal{S}_0$ and sends $q_1 = q_0 \oplus I$ to $\mathcal{S}_1$, where
$I \in {\{0,1\}}^N$ has value 1 at index $i$ and value 0 everywhere else.
Server $\mathcal{S}_0$ computes a response $r_0$ as
\[\bigoplus_{1 \leq j \leq N : q_0[j] = 1}D[j].\]
Similarly, server $\mathcal{S}_1$ computes a response $r_1$ as
\[\bigoplus_{1 \leq j \leq N : q_1[j] = 1}D[j].\]
\begingroup%
\setlength{\intextsep}{0pt}%
\setlength{\columnsep}{0pt}%
\begin{wrapfigure}{r}{0.51\textwidth}
  \begin{center}
    \begin{tabular}{c}
    \begin{lstlisting}
      $\Rand{q_0}{{\{0,1\}}^N};$
      $\Assn{q_1}{q_0 \oplus I};$
      $\DFor{i}{1}{N}{}$
      $\quad \RCondt{q_0[i]=1}{\Assn{a_0[i]}{D[i]}};$
      $\quad \RCondt{q_1[i]=1}{\Assn{a_1[i]}{D[i]}};$
      $\Assn{r_0}{\bigoplus_{j \in [1,N]}a_0[j]};$
      $\Assn{r_1}{\bigoplus_{j \in [1,N]}a_1[j]};$
      $\Assn{v}{r_0 \oplus r_1};$
    \end{lstlisting}
    \end{tabular}
  \end{center}
  \caption{Private information retrieval (PIR) scheme}\label{fig:pir}    
\end{wrapfigure}%
Finally, $\mathcal{C}$ computes $r_0 \oplus r_1 = D[i]$. Since it is assumed
that $\mathcal{S}_0$ and $\mathcal{S}_1$ do not collude, the fact that the
queries $q_0$ and $q_1$ are each uniformly distributed ensure that no
information about the index $i$ is leaked. The fact that $q_0$ and $q_1$ are not
independent, however, means that the protocol does not ensure secrecy in the
presence of collusion---information about $i$ may (and in this case, does) leak
out through different correlations between $q_0$ and $q_1$.

The combined program in \cref{fig:pir} models this protocol (arrays
$a_0,a_1$ initialized with 0). We establish security by proving two
different properties of the program in our logic: uniformity and
input independence. While both properties amount
to the same security property, their proofs are different and
demonstrate our logic's flexibility.

\endgroup

\subsubsection{Proof of Uniformity}

Starting from the trivial pre-condition $\Phi_1 \triangleq \top$, we
would like to prove the post-condition
\[
  \Psi \triangleq \underbrace{\Unif[q_0]}_{\mathclap{\text{$\mathcal{S}_0$'s view}}} \land \underbrace{\Unif[q_1]}_{\mathclap{\text{$\mathcal{S}_1$'s view}}}.
\]
This says that the views of $\mathcal{S}_0$ and $\mathcal{S}_1$ ($q_0$ and
$q_1$, respectively) are uniformly random bitstrings.

By \rname{RSamp}, adjoining the sampling for  $q_0$ (line 1) gives
\[
  \Unif[q_0].
\]
Since $I$ is a deterministic variable, we can adjoin $\Dist[I]$:
\[
  \Unif[q_0] \sepand \Dist[I].
\]
By \rname{RAssn*}, assigning to $q_1$ (line 2) gives
\[
  \Unif[q_0] \sepand \Dist[I] \mathrel{\land} q_1 \sim q_0 \mathrel{\oplus} I.
\]
Next, we can pull out $\Unif[q_0]$ like so:
\[
  \Unif[q_0] \mathrel{\land} (\Unif[q_0] \sepand \Dist[I] \mathrel{\land} q_1 \sim q_0 \mathrel{\oplus} I).
\]
We then apply the xor axiom (\ref{ax:xor}) to the right conjunct, giving the
desired post-condition:
\[
  \Psi \triangleq \Unif[q_0] \mathrel{\land} \Unif[q_1].
\]
Since $q_0$ and $q_1$ are unmodified in the remainder of the program, we can
preserve $\Psi$ through to the end using \rname{Const} and \rname{True}.

\subsubsection{Proof of Input Independence}

Starting from the pre-condition $\Phi_1 \triangleq \Dist[I]$, now treating $I$
as a randomized variable, we would like to prove the post-condition
\[
  \Psi \triangleq \underbrace{\Dist[I]}_{\text{Index}} \sepand \underbrace{\Dist[q_0]}_{\mathclap{\text{$\mathcal{S}_0$'s
  view}}} \mathrel{\land} \underbrace{\Dist[I]}_{\text{Index}} \sepand \underbrace{\Dist[q_1]}_{\mathclap{\text{$\mathcal{S}_1$'s
  view}}}.
\]
This says that the views of $\mathcal{S}_0$ and $\mathcal{S}_1$ ($q_0$ and
$q_1$, respectively) are independent of secret index $I$.

By \rname{RSamp*}, adjoining the sampling for $q_0$ (line 1) gives
\[
  \Unif[q_0] \sepand \Dist[I].
\]
By \rname{RAssn*}, assigning to $q_1$ (line 2) gives
\[
  \Unif[q_0] \sepand \Dist[I] \mathrel{\land} q_1 \sim q_0 \mathrel{\oplus} I.
\]
Next, we can pull out $\Dist[I] \sepand \Unif[q_0]$ like so:
\[
  \Dist[I] \sepand \Unif[q_0] \mathrel{\land}
  (\Unif[q_0] \sepand \Dist[I] \mathrel{\land} q_1 \sim q_0 \mathrel{\oplus} I).
\]
We then apply the xor axiom (\ref{ax:xor}) to the right conjunct, giving
\[
  \Dist[I] \sepand \Unif[q_0] \mathrel{\land} \Dist[I] \sepand \Unif[q_1],
\]
which implies the desired post-condition:
\[
  \Psi \triangleq \Dist[I] \sepand \Dist[q_0] \mathrel{\land} \Dist[I] \sepand \Dist[q_1].
\]
Since $q_0$ and $q_1$ are unmodified in the remainder of the program, we can
preserve $\Psi$ through to the end using \rname{Const} and \rname{True}.

\subsection{Oblivious Transfer}

\begin{figure}
  \begin{lstlisting}
    $\Rand{r_0,r_1}{{\{0,1\}}^k};$
    $\Rand{d}{\{0,1\}};$
    $\RCond{d = 0}{\Seq{\Assn{r_d}{r_0}}{\textcolor{gray}{\Assn{r_{1-d}}{r_1}}}}{\Seq{\Assn{r_d}{r_1}}{\textcolor{gray}{\Assn{r_{1-d}}{r_0}}}};$
    $\Assn{e}{c \oplus d};$
    $\RCond{e = 0}{\Seq{\Assn{f_0}{m_0 \oplus r_0}}{\Assn{f_1}{m_1 \oplus r_1}}}
    {\Seq{\Assn{f_0}{m_0 \oplus r_1}}{\Assn{f_1}{m_1 \oplus r_0}}};$
    $\DCond{c = 0}{\Seq{\Assn{m_c}{f_0 \oplus r_d}}{\textcolor{gray}{\Assn{f_{1-c}}{m_1 \oplus
  r_{1-d}}}}}{\Seq{\Assn{m_c}{f_1 \oplus r_d}}{\textcolor{gray}{\Assn{f_{1-c}}{m_0 \oplus
  r_{1-d}}}}};$
  \end{lstlisting}
  \caption{Oblivious transfer (OT) scheme}\label{fig:ot}  
\end{figure}

Oblivious transfer (OT) is a common building block in many cryptographic
protocols~\citep{rabin2005exchange}. It involves two parties: a sender
$\mathcal{S}$ holding two secrets $m_0,m_1$ and a receiver $\mathcal{R}$ holding
a \emph{choice} bit $c \in \{0,1\}$. Through the protocol, $\mathcal{R}$ learns
the secret $m_c$ but nothing about $m_{1-c}$, while $\mathcal{S}$ learns nothing
about $c$. If the setup can be performed by a trusted third party $\mathcal{T}$,
the following simple protocol implements OT~\citep{rivest1999unconditionally}:
\begin{enumerate}
  \item $\mathcal{T}$ sends $\mathcal{S}$ two random $k$-bit strings $r_0,r_1$.
  \item $\mathcal{T}$ sends $\mathcal{R}$ a random bit $d$ and the string $r_d$.
  \item $\mathcal{R}$ sends $\mathcal{S}$ the value $e = c \oplus d$.
  \item $\mathcal{S}$ sends $\mathcal{R}$ the values $f_0 = m_0 \oplus r_e$ and
  $f_1 = m_1 \oplus r_{1-e}$.
  \item $\mathcal{R}$ computes $m_c = f_c \oplus r_d$.
\end{enumerate}
As part of the trusted setup in steps 1 and 2, $\mathcal{T}$ essentially hands
$\mathcal{S}$ and $\mathcal{R}$ one-time pad keys to encrypt their secrets.
$\mathcal{R}$ uses $d$ to encrypt its choice $c$ (step 3), $\mathcal{S}$ uses
$r_0,r_1$ to encrypt its secrets $m_0,m_1$ using $\mathcal{R}$'s message to pick
which key to use for which message (step 4), and finally $\mathcal{R}$ uses
$r_d$ to decrypt one of the secrets (step 5).

The protocol ensures perfect secrecy for $\mathcal{R}$'s choice $c$, since $e$
is an encryption of $c$ under the OTP with key $d$, which is kept secret from
$\mathcal{S}$. The protocol ensures perfect secrecy for one of $\mathcal{S}$'s
secrets $m_0,m_1$, since $f_0,f_1$ are encryptions of $m_0,m_1$ under the OTP
with keys $r_e,r_{1-e}$, and the fact that $\mathcal{R}$ is given one of the
keys ($r_d$) means that it can decrypt exactly one of $f_0$ or $f_1$.

The combined program in \cref{fig:ot} models the OT protocol. Stating security
for $\mathcal{S}$ requires a bit of work. We first instrument the program with
ghost code, shown in gray. Intuitively, the ghost code computes the encrypted
version of the wrong message, i.e., the one that $\mathcal{R}$ did \emph{not}
request---the combined view of $\mathcal{R}$ should then be uniform.

\subsubsection{Proof of Uniformity}

Starting from the trivial pre-condition $\Phi_1 \triangleq \top$, we would like
to prove the post-condition
\[
  \Psi \triangleq \underbrace{(\Unif_{k \times
  k}[(r_0,r_1)] \sepand \Unif[e])}_{\mathclap{\text{$\mathcal{S}$'s
  view}}} \mathrel{\land} \underbrace{(\Unif[d] \sepand \Unif_{k \times
  k}[(r_d,f_{1-c})])}_{\mathclap{\text{$\mathcal{R}$'s ``view''}}}.
\]
To establish $\mathcal{R}'s$ secrecy, we need to consider the view of
$\mathcal{S}$, which consists of $r_0$, $r_1$, and $e$. For $\mathcal{R}$'s
choice $c$ to be kept secret, it is required that $\Unif_{k \times
k}[(r_0,r_1)] \sepand \Unif[e]$, i.e., $\mathcal{S}$'s combined view is
uniform. Note that it is not enough to establish that the individual components
of $\mathcal{S}$'s view are uniform. To see why, suppose $\mathcal{T}$ also
sends $\mathcal{S}$ the random bit $d$, which reveals $c = e \oplus d$. Although
the individual components of $\mathcal{S}$'s view would indeed be uniform, i.e.,
$\Unif_k[r_0] \mathrel{\land} \Unif_k[r_1] \mathrel{\land} \Unif[e] \mathrel{\land}
\Unif[d]$, $\mathcal{R}$'s secrecy is clearly violated. The stronger
post-condition establishes that $\mathcal{S}$'s combined view is uniform.

To establish $\mathcal{S}'s$ (one-sided) secrecy, we need to consider the view
of $\mathcal{R}$, which consists of $d$, $r_d$, and one of $f_0$ or $f_1$. In
particular, the $f_i$ to be considered corresponds to the encryption of the
``wrong'' message, which we assign to the ghost variable $f_{1-c}$ (which is, in
turn, computed using the ghost variable $r_{1-d}$). Similar to $\mathcal{R}$'s
secrecy, it is then required that $\Unif[d] \sepand \Unif_{k \times
k}[(r_d,f_{1-c})]$, i.e., $\mathcal{R}$'s combined view is uniform.

We first show $\mathcal{R}$'s secrecy, followed by $\mathcal{S}$'s secrecy, and
then combine the results using \rname{Conj}.  By \rname{RSamp}
and \rname{RSamp*}, we can adjoin the random samplings for
$r_0,r_1,d$ (lines 1--2), giving
\[
  \Unif_k[r_0] \sepand \Unif_k[r_1] \sepand \Unif[d].
\]
Since the free variables of this formula are unmodified in the conditional (line
3), we can preserve the formula using \rname{Const} and \rname{True}.
Since $c$ is a deterministic variable we can adjoin $\Dist[c]$, giving
\[
  \Dist[c] \sepand \Unif_k[r_0] \sepand \Unif_k[r_1] \sepand \Unif[d].
\]
For the assignment to $e$ (line 4), we start from the local pre-condition
\[
  \Dist[c] \sepand \Unif[d].
\]
By \rname{RAssn*}, assigning to $e$ gives
\[
  (\Dist[c] \sepand \Unif[d]) \mathrel{\land} e \sim c \mathrel{\oplus} d.
\]
Applying the xor axiom (\ref{ax:xor}) leaves
\[
  \Unif[e].
\]
Then, we can frame as follows:
\begin{mathpar}
    \hspace*{0.9cm}\inferrule*[Left=Frame]
    { \vdash \psl{\phi}{c'}{\psi} \quad\,\,\,
      \FV(\eta) \cap \MV(c') = \emptyset \quad\,\,\,
      \FV(\psi) \subseteq \FV(\phi) \cup \WV(c') \quad\,\,\,
    \models \phi \to \Dist[\RV(c')]}
    { \vdash \psl{\underbrace{\Dist[c] \sepand \Unif[d]}_\phi \sepand \underbrace{\Unif_k[r_0] \sepand \Unif_k[r_1]}_\eta}{\underbrace{\Assn{e}{c \oplus
    d}}_{c'}}{\underbrace{\Unif[e]}_\psi \sepand \underbrace{\Unif_k[r_0] \sepand \Unif_k[r_1]}_\eta}
    }.
\end{mathpar}
The post-condition implies
\[
  \Unif_{k \times k}[(r_0,r_1)] \sepand \Unif[e],
\]
which establishes $\mathcal{R}$'s secrecy. Since $r_0$, $r_1$, and $e$ are
unmodified in the remainder of the program, we can preserve $\Unif_{k \times
k}[(r_0,r_1)] \sepand \Unif[e]$ through to the end using \rname{Const}.

Next, we show $\mathcal{S}$'s secrecy. Again, by \rname{RSamp}
and \rname{RSamp*}, we can adjoin the random samplings for $r_0,r_1,d$ (lines
1--2), giving
\[
  \Unif_k[r_0] \sepand \Unif_k[r_1] \sepand \Unif[d].
\]
We go through the conditional (line 3) with \rname{RCond}, which gives
pre-condition
\[
  \Unif_k[r_0] \sepand \Unif_k[r_1] \sepand d = 0 \sim \ktt.
\]
To go through the first assignment, we start from the local pre-condition
\[
  \Unif_k[r_0].
\]
By \rname{RAssn*}, assigning to $r_d$ gives
\[
  \Unif_k[r_0] \mathrel{\land} r_d \sim r_0.
\]
Transferring the distribution law gives
\[
  \Unif_k[r_d].
\]
We can then frame in $\Unif_k[r_1] \sepand d = 0 \sim \ktt$ giving
\[
  \Unif_k[r_d] \sepand \Unif_k[r_1] \sepand d = 0 \sim \ktt.
\]
The second assignment follows similarly, this time starting from the local
pre-condition
\[
  \Unif_k[r_1]
\]
and giving the following post-condition in the true branch
\[
  \Unif_k[r_d] \sepand \Unif_k[r_{1-d}] \sepand d = 0 \sim \ktt.
\]
The false branch yields the same post-condition, which, by \rname{RCond},
brings us to the post-condition
\[
  \Unif_k[r_d] \sepand \Unif_k[r_{1-d}] \sepand \Unif[d].
\]
Since the free variables of this formula are unmodified in lines 4--5, we can
preserve this formula through using \rname{Const} and \rname{True}. Next, we go
through the deterministic conditional (line 6) using \rname{DCond}. In the true
branch, we start with pre-condition
\[
  \Unif_k[r_d] \sepand \Unif_k[r_{1-d}] \sepand \Unif[d] \mathrel{\land} (c=0)
  = \ktt.
\]
Dropping the right conjunct, we can adjoin $\Dist[m_1]$ like so
\[
  \Unif_k[r_d] \sepand \Dist[m_1] \sepand \Unif_k[r_{1-d}] \sepand \Unif[d],
\]
since $m_1$ is a deterministic variable.  We preserve this formula through the
first assignment to $m_c$ using \rname{Const} and \rname{True}, and then go
through the second assignment to $f_{1-c}$ starting from the local pre-condition
\[
  \Dist[m_1] \sepand \Unif_k[r_{1-d}].
\]
Applying \rname{RAssn*} and the xor axiom (\ref{ax:xor}) gives
\[
  \Unif_k[f_{1-c}].
\]
Framing in $\Unif_k[r_d] \sepand \Unif[d]$ gives the following post-condition in
the true branch
\[
  \Unif_k[r_d] \sepand \Unif_k[f_{1-c}] \sepand \Unif[d].
\]
The false branch yields the same post-condition. Then, we can merge
$\Unif_k[r_d] \sepand \Unif_k[f_{1-c}]$ and rearrange like so
\[
    \Unif[d] \sepand \Unif_{k \times k}[(r_d,f_{1-c})].
\]
This establishes $\mathcal{S}$'s secrecy. Combining the formulas establishing
$\mathcal{R}$'s secrecy and $\mathcal{S}$'s secrecy using \rname{Conj} gives the
desired post-condition.

\subsubsection{Proof Attempt: Input Independence}

Again, we can try to prove security via input independence instead of
uniformity. Starting from the pre-condition
\[
  \Phi_1 \triangleq \Dist[c] \mathrel{\land} \Dist[m_0] \mathrel{\land} \Dist[m_1],
\]
we would like to prove the post-condition
\[
  \Psi \triangleq (\underbrace{\Dist[c]}_{\mathclap{\text{$\mathcal{R}$'s
  secret choice}}} \sepand \underbrace{(\Dist[r_0] \mathrel{\land} \Dist[r_1] \mathrel{\land} \Dist[e])}_{\mathclap{\text{$\mathcal{S}$'s
  view}}}) \mathrel{\land}
  (\underbrace{\Dist[m_{1-c}]}_{\mathclap{\text{$\mathcal{S}$'s
  unselected secret}}} \sepand \underbrace{(\Dist[d] \mathrel{\land} \Dist[r_d] \mathrel{\land} \Dist[f_0] \mathrel{\land} \Dist[f_1])}_{{\text{$\mathcal{R}$'s view}}}).
\]
To establish $\mathcal{R}'s$ secrecy, we need to show that the secret choice $c$
is independent from $\mathcal{S}$'s view. Similarly, to establish
$\mathcal{S}$'s (one-sided) secrecy, we need to show that the unselected secret,
which we assign to the ghost variable $m_{1-c}$ is independent from
$\mathcal{R}$'s view.

However, here we run into difficulties---it does not seem possible to prove this
judgment in our logic, and even sketching a proof on paper is not easy. In
general, stating and proving perfect security as input independence is trickier
when there are multiple parties, like in OT\@. Investigating how to prove this
kind of property is an interesting direction for further work.

\subsection{Multi-Party Computation}

\begingroup
\setlength{\intextsep}{0pt}
\setlength{\columnsep}{0pt}
\begin{wrapfigure}{r}{0.54\textwidth}
  \begin{center}
    \begin{tabular}{c}
    \begin{lstlisting}
      $\DFor{i}{1}{3}{}$
      $\quad \Rand{r[i].1}{\mathbb{Z}_p};$
      $\quad \Rand{r[i].2}{\mathbb{Z}_p};$
      $\quad \Assn{r[i].3}{x[i] - r[i].1 - r[i].2 \mod p};$
      $\DFor{i}{1}{3}{}$
      $\quad \Assn{s[i]}{r[1].i + r[2].i + r[3].i \mod p};$
      $\Assn{v}{s[1] + s[2] + s[3] \mod p};$
    \end{lstlisting}
    \end{tabular}
  \end{center}
  \caption{Three-party secure addition protocol}\label{fig:mpc}  
\end{wrapfigure}

Secure multi-party computation (MPC) allows mutually untrusting parties to
jointly compute a function of their private inputs without revealing
them~\citep{yao1986generate,goldreich1987play}. The parties agree on a function
$f$ and then use an MPC protocol to securely compute $v = f(x_1,\ldots,x_n)$,
where $x_i$ is party $P_i$'s private input. MPC guarantees that parties learn
$v$, and nothing more.

As an example, we consider secure computation of addition, i.e., of the function
$f(x_1,\ldots,x_n) = \sum_{i=1}^{n} x_i$~\citep{cramer2015secure}. This simple
function turns out to be surprisingly useful, for example, for privately
totaling salaries of employees in a company or votes in secure electronic
voting. Secure addition can be achieved by the following simple protocol with
inputs $x_i \in \mathbb{Z}_p$, where $p$ is a fixed prime number agreed upon in
advance. We describe the three-party case for simplicity, but the protocol
easily extends to $n$ parties.

\endgroup

\begin{enumerate}
  \item Each $P_i$ encodes their input as three \emph{secret shares} by choosing
    $r_{i,1},r_{i,2}$ uniformly at random in $\mathbb{Z}_p$ and setting $r_{i,3}
    = x_i - r_{i,1} - r_{i,2} \mod p$.
  \item Each $P_i$ sends $r_{i,2},r_{i,3}$ to $P_1$, $r_{i,1},r_{i,3}$ to $P_2$,
    and $r_{i,1},r_{i,2}$ to $P_3$.
  \item Each $P_j$ computes the sum $s_\ell = r_{1,\ell} + r_{2,\ell} +
    r_{3,\ell} \mod p$ for $\ell \neq j$ and sends $s_\ell$ to all parties.
  \item All parties compute the result $v = s_1 + s_2 + s_3 \mod p$.
\end{enumerate}
The \emph{end-to-end} security of the protocol, i.e., that parties learn no new
information beyond the output $v$, is subtle to prove. In a nutshell, the
security of the protocol is usually established by
\emph{simulation}~\citep{lindell2017simulate}, a proof technique that is
pervasive in cryptography but does not have a clean translation to our logic.
The interested reader should see the monograph by~\citet{cramer2015secure}.

The security of secret sharing, however, is expressible in our logic.
Informally, the protocol splits each secret into three pieces (``shares''), and
the security property ensures that knowing at most two of the three shares
reveals no information about the secret. In step (1), each party generates
shares of their secret input $x_i$ by selecting $r_{i,1},r_{i,2},r_{i,3}$
uniformly at random from $\mathbb{Z}_p$, subject to the constraint that the
shares add up to $x_i$. In step (2), each party distributes secret shares in
such a way that no other party learns any information about their secret input.
For concreteness, consider $P_2$'s view: It knows the values $r_{i,1}$ and
$r_{i,3}$, and that $x_i = r_{i,1} + r_{i,2} + r_{i,3} \mod p$, but since
$r_{i,2}$ is chosen uniformly from $\mathbb{Z}_p$, any value of $x_i$ is equally
likely. Thus, no information about $x_i$ is leaked.

The combined program in \cref{fig:mpc} models the secure addition protocol; the
secret sharing steps correspond to lines 1--4.  Like in our previous examples,
we prove the security of secret sharing in two ways: by establishing uniformity
and input independence.

\subsubsection{Proof of Uniformity}
Starting from the trivial pre-condition $\Phi_1 \triangleq \top$, we would like
to prove the post-condition
\[
  \Psi \triangleq \bigwedge_{\alpha \in \{2,3\}} \underbrace{\Unif[(r[\alpha].2,
  r[\alpha].3)]}_{\mathclap{\text{$P_1$'s view from $P_\alpha$}}} \mathrel{\land}
  \bigwedge_{\alpha \in \{1,3\}} \underbrace{\Unif[(r[\alpha].1,
  r[\alpha].3)]}_{\mathclap{\text{$P_2$'s view from
  $P_\alpha$}}}  \mathrel{\land}
  \bigwedge_{\alpha \in \{1,2\}} \underbrace{\Unif[(r[\alpha].1, r[\alpha].2)]}_{\mathclap{\text{$P_3$'s view from $P_\alpha$}}}.
\]
This says that each party's view from the other parties is uniform and independent.

To prove this post-condition, we take the following for-loop invariant:
\[
  \bigsep_{\alpha \in [1, i)} \Unif[(r[\alpha].1,r[\alpha].2)] \mathrel{\land} \Unif[(r[\alpha].2,r[\alpha].3)] \mathrel{\land} \Unif[(r[\alpha].1,r[\alpha].3)].
\]
By \rname{RSamp*}, adjoining the random samplings for $r[i].1$ and $r[i].2$
(lines 2 and 3) gives
\[
  \Unif[r[i].1] \sepand \Unif[r[i].2] \sepand \bigsep_{\alpha \in [1,
    i)} \Unif[(r[\alpha].1,r[\alpha].2)] \mathrel{\land} \Unif[(r[\alpha].2,r[\alpha].3)] \mathrel{\land} \Unif[(r[\alpha].1,r[\alpha].3)].
\]
To go through the assignment to $r[i].3$ (line 4), we start from the local pre-condition
\[
  \Unif[r[i].1] \sepand \Unif[r[i].2].
\]
By \rname{RAssn*}, assigning to $r[i].3$ gives
\[
  \Unif[r[i].1] \sepand \Unif[r[i].2] \mathrel{\land} r[i].3 \sim x[i] - r[i].1
  - r[i].2 \mod p.
\]
Applying the modular addition axiom (\ref{ax:mod}) and merging pairwise
independent assertions gives
\[
  \Unif[(r[i].1,r[i].2)] \mathrel{\land} \Unif[(r[i].2,r[i].3)] \mathrel{\land}
  \Unif[(r[i].1,r[i].3)].
\]
Framing in the invariant for the earlier iterations establishes the loop
invariant in \rname{DFor}, giving:
\[
  \bigsep_{\alpha \in
  [1,3]} \Unif[(r[\alpha].1,r[\alpha].2)] \mathrel{\land} \Unif[(r[\alpha].2,r[\alpha].3)] \mathrel{\land} \Unif[(r[\alpha].1,r[\alpha].3)].
\]
After rearranging and dropping terms, this formula implies the desired
post-condition $\Psi$.  Since the free variables of $\Psi$ are unmodified in the
remainder of the program, we can preserve it through to the end
using \rname{Const} and \rname{True}, establishing uniformity.

\subsubsection{Proof of Input Independence}

Starting from the pre-condition
\[
  \Phi_1 \triangleq \bigwedge_{\alpha \in [1,3]} \Dist[x[\alpha]],
\]
we would like to prove the post-condition
\begin{align*}
  \Psi \triangleq \bigwedge_{\alpha \in \{2,3\}} \underbrace{\Dist[x[\alpha]]}_{\mathclap{\text{$P_\alpha$'s
  input}}} \sepand \underbrace{\Dist[(r[\alpha].2,
  r[\alpha].3)]}_{\mathclap{\text{$P_1$'s view from
  $P_\alpha$}}} \mathrel{\land} &\bigwedge_{\alpha \in \{1,3\}} \underbrace{\Dist[x[\alpha]]}_{\mathclap{\text{$P_\alpha$'s
  input}}} \underbrace{\sepand \Dist[(r[\alpha].1,
  r[\alpha].3)]}_{\mathclap{\text{$P_2$'s view from
  $P_\alpha$}}} \mathrel{\land}\\ &\bigwedge_{\alpha \in \{1,2\}} \underbrace{\Dist[x[\alpha]]}_{\mathclap{\text{$P_\alpha$'s
  input}}} \sepand \underbrace{\Dist[(r[\alpha].1,
  r[\alpha].2)]}_{\mathclap{\text{$P_3$'s view from $P_\alpha$}}}.
\end{align*}
This says that, for each party, the secret input of each other party is
independent from the view they generate. We defer details of this proof
to \fullref{app:examples}.


\subsection{Simple Oblivious RAM}

Consider a core programming language defined by the following syntax:
\[
  P ::= \epsilon \mid i;P \quad \mbox{where} ~~~ i::=\iread{x} \mid \iwrite{x}{v}
\]
with $x$ ranging over a set $\mathcal{X}$ of registers and $v$ ranging over
integers $\mathbb{Z}$. A simple execution model for this language is random
access memory (RAM). Informally, a RAM machine maintains a partial mapping from
registers to integers, a program counter that tracks which instruction is to be
executed next, and reads and updates the mapping according to the program
instructions.

\emph{Oblivious RAM} (ORAM)~\citep{stoc/Goldreich87,GoldreichO96} is a
probabilistic execution model guaranteeing that an adversary who observes the
sequence of accessed memory locations---but not their contents---only learns the
length of the program. The basic idea of ORAM is to maintain a mapping from
logical addresses accessed by the client program to physical addresses where
data is stored; this mapping is re-scrambled after each read and write. We
consider a simple and idealized variant of ORAM inspired by
\citet{iacr/ChungP13a}, and use our logic to prove its security.

\subsubsection{Definition of Simple ORAM}
ORAM assumes a memory model split into two parts: an \emph{external}, insecure
bulk memory where accesses are visible to the adversary, and an \emph{internal},
secure memory with a small number of registers where accesses are not visible to
the adversary. In our variant, the internal memory stores a so-called
\emph{position map}, while the external memory is organized into a tree where
values can be read and written. We treat the external memory $a$ as a map from
addresses to buckets. Addresses are bitstrings of length at most $n$; they are
partially ordered by the prefix relation, and are thus structured as a tree.
Leaf addresses are bitstrings of length exactly $n$. Each node of the tree
stores a \emph{bucket} a list of triples of the form $(x,v,l)$, where $x
\in\mathcal{X}$ is a register, $v\in\mathbb{Z}$ is an integer, and
$l\in{\{0,1\}}^n$ is a leaf address. One key invariant of the ORAM scheme is that
$(x,v,l)$ will be stored in some bucket along the path from root address
$\epsilon$ to leaf address $l$, i.e.\, in a bucket $a[i]$ at some address $i$
that is a prefix of $l$. The position map $p$ is a mapping from registers $x \in
\mathcal{X}$ to leaf addresses; this mapping ensures obliviousness by
introducing a level of indirection.

\begin{figure*}[t!]
    \centering
    \begin{subfigure}[t]{0.3\textwidth}
        \centering
        \includegraphics[height=1.4in]{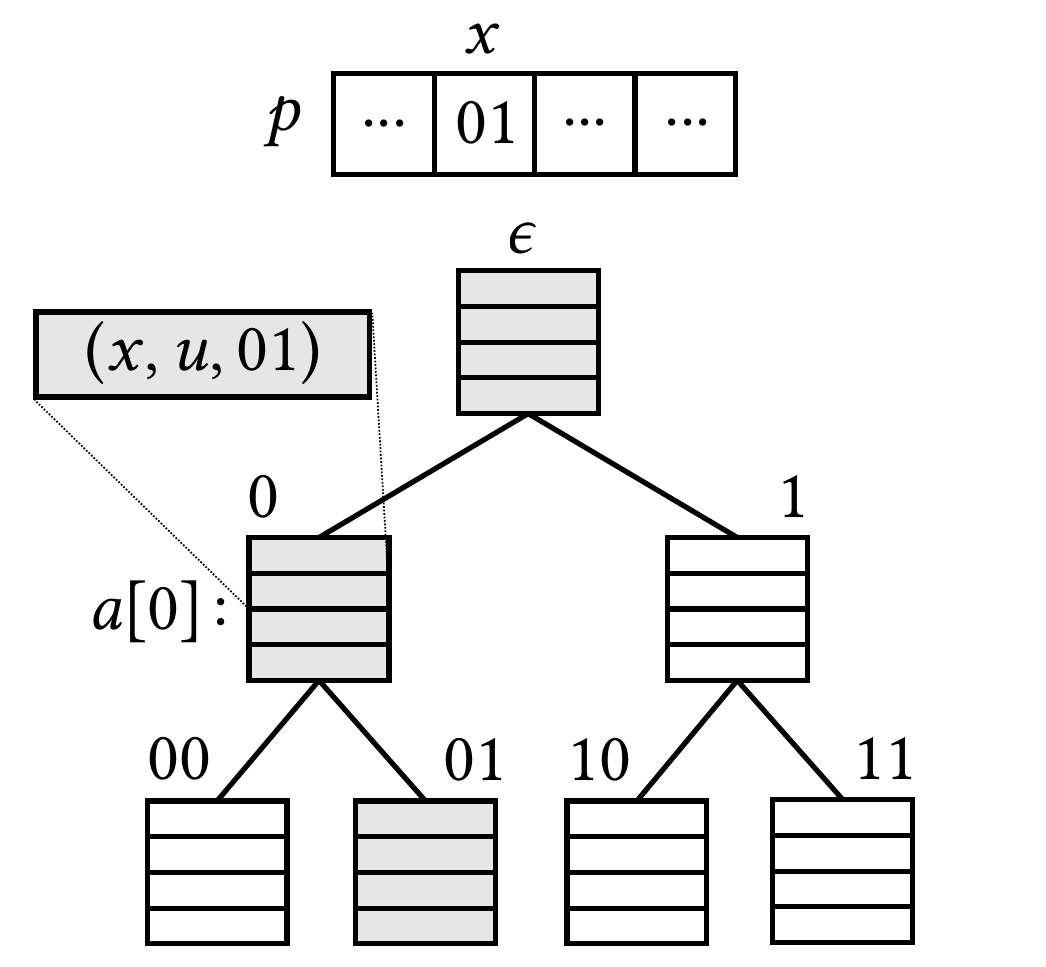}
        \caption{$\mathsf{read}(x)$}
    \end{subfigure}%
    \begin{subfigure}[t]{0.3\textwidth}
        \centering
        \includegraphics[height=1.4in]{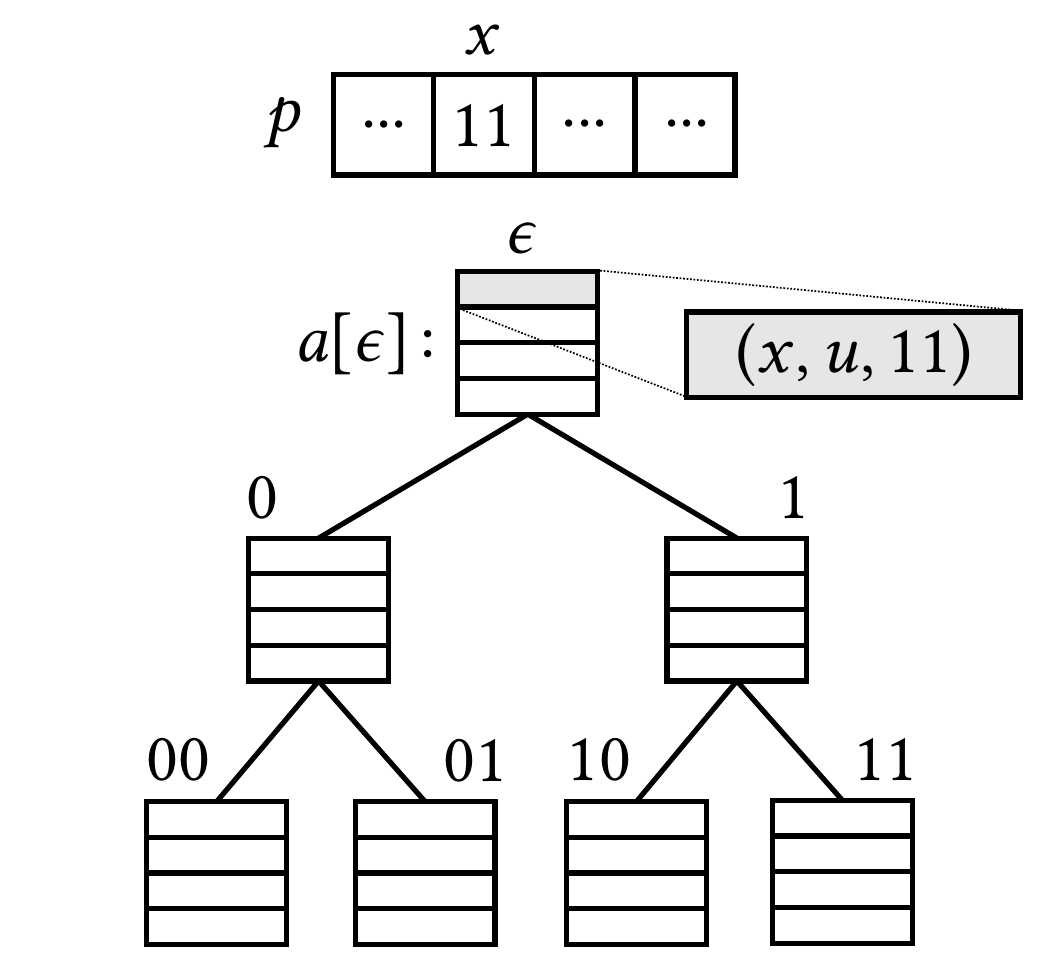}
        \caption{Update $p[x]$}
    \end{subfigure}%
    \begin{subfigure}[t]{0.3\textwidth}
        \centering
        \includegraphics[height=1.4in]{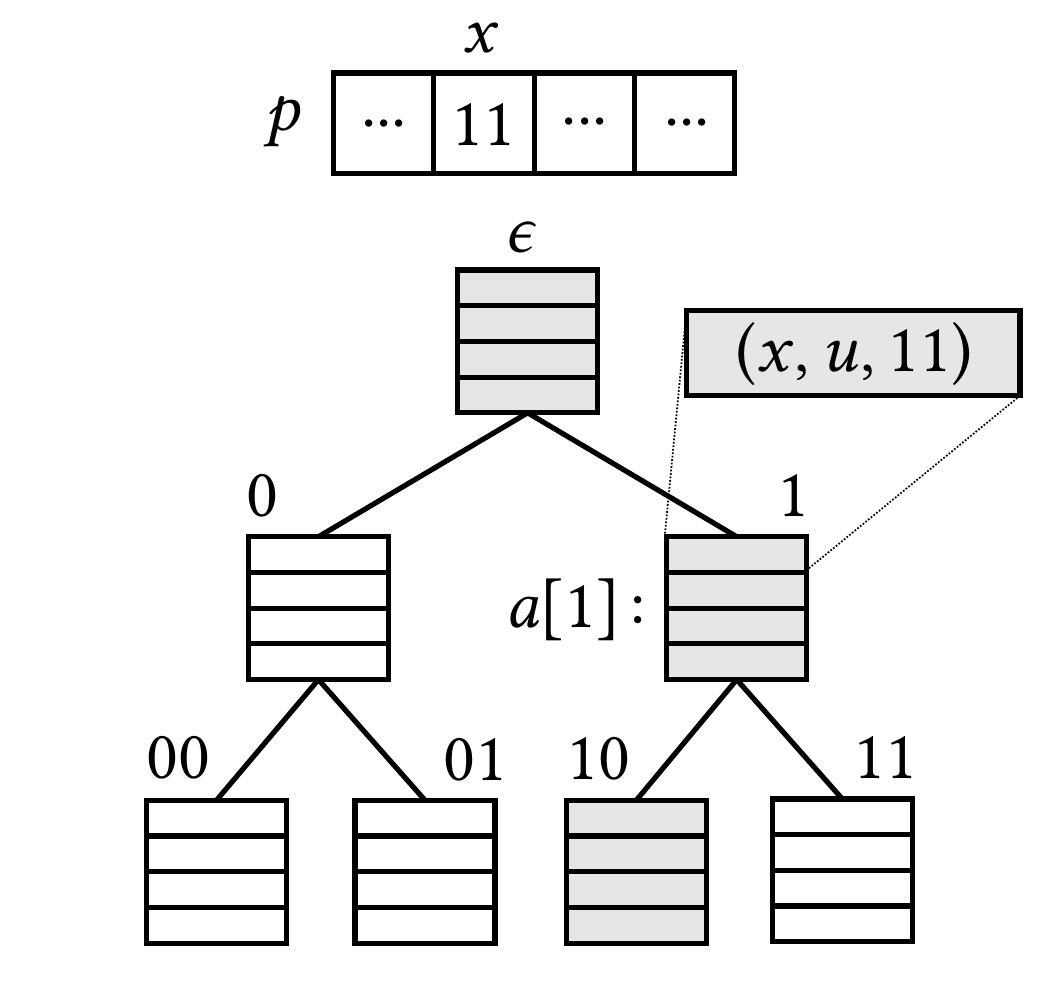}
        \caption{$\mathsf{flush}$}
    \end{subfigure}%
    \caption{Simple ORAM}\label{fig:soram}
\end{figure*}

We briefly describe the semantics of the $\iread{x}$ instruction;
\cref{fig:soram} illustrates an example. To retrieve the value of $x$, we read
$p[x]$ from the position map and then search through all buckets along the path
from $\epsilon$ to $p[x]$ for a triple of the form $(x,u,p[x])$
(\cref{fig:soram} (a)). Once the triple is found, it is removed from its bucket.
In order to guarantee obliviousness, it is necessary to read all addresses along
the path to the leaf, regardless of where the target entry is found. Once the
leaf is reached, we sample a fresh bitstring $l$ of length $n$, add $(x,u,l)$ to
the bucket $a[\epsilon]$, and then update $p[x]$ to hold $l$ (\cref{fig:soram}
(b)). The semantics of the $\iwrite{x}{v}$ instruction is similar, except that
the updated entry $(x,v,l)$ is added to the bucket $a[\epsilon]$.

While the operations so far ensure obliviousness, all triples will accumulate at
the root of the tree, i.e., in the bucket $a[\epsilon]$. To better balance the
buckets, each \textsf{read} or \textsf{write} operation is followed by a
\textsf{flush} operation, which samples another bitstring $l$ of length $n$ and
then traverses the tree from root $\epsilon$ to leaf $l$ while pushing every
triple $(x,v,l')$ along its path as far down as possible, namely, to bucket
$a[\lcp(l, l')]$ where $\lcp(l, l')$ is the longest common prefix of $l$ and
$l'$ (\cref{fig:soram} (c)).

\cref{fig:oram} defines the oblivious semantics of core language programs by
compilation to {\PWHILE}. To model the adversary's view of the accesses, the
compilation instruments the code to store \emph{leakage information} in the
variable $\ell$, as explained below. We briefly comment on the notation and
operators. We use ${[]}$ for the empty list, ${::}$ for adding an element to a list,
and $+$ for concatenating two lists. Given a bitstring $i$ of length $n$ and
$k\leq n$, we let $i[1, \ldots, k]$ be the bitstring of length $k$ consisting of
the first $k$ bits of $i$; we let $i[1, \ldots, 0] = \epsilon$ denote the empty
bitstring. The operator $\mathsf{split}_{\{ (x,v,l) \mid \phi \}}$ iterates over
a list of triples and returns two sublists of elements satisfying $\phi$ and
elements not satisfying $\phi$ respectively, where $\phi$ can mention $(x, v, l)$.

\begin{figure}[t!]
  \begin{subfigure}[t]{0.5\textwidth}
  \begin{center}
  \underline{Compilation of \textsf{read}}    
  \begin{tabular}{c}
  \Suppressnumber%
  \begin{lstlisting}[firstnumber=1,numbersep=-3pt]
    $\mbox{}\hspace{-1cm} \mathcal{C}(\iread{x})=$|\Reactivatenumber|
    ${\textcolor{gray}{\Assn{\ell[c].1}{p[x]}}};$
    $\Assn{w}{[]};$
    $\DFor{j}{1}{n}{}$
    $\quad \Assn{(w',a[p[x][1, \ldots, j]])}{}$
    $\quad\quad \mathsf{split}_{\{ (y,-,-) \mid x= y\}} (a[p[x][1, \ldots, j]]);$
    $\quad \Assn{w}{w+w'};$
    $\Assn{(y,u,l)}{\mathsf{head}(w)};$
    $\Rand{p[x]}{{\{0,1\}}^n};$
    $\Assn{a[\epsilon]}{(x,u,p[x])::a[\epsilon]}$
  \end{lstlisting}
  \end{tabular}
  \end{center}   
  \end{subfigure}%
  \begin{subfigure}[t]{0.5\textwidth}
  \begin{center}
  \underline{Compilation of \textsf{write}}
  \begin{tabular}{c}
  \Suppressnumber%
  \begin{lstlisting}[firstnumber=1,numbersep=-3pt]
    $\mbox{}\hspace{-1cm} \mathcal{C}(\iwrite{x}{v})=$|\Reactivatenumber|
    ${\textcolor{gray}{\Assn{\ell[c].1}{p[x]}}};$
    $\Assn{w}{[]};$    
    $\DFor{j}{1}{n}{}$
    $\quad \Assn{(w',a[p[x][1, \ldots, j]])}{}$
    $\quad\quad \mathsf{split}_{\{ (y,-,-) \mid x= y\}} (a[p[x][1, \ldots, j]]);$
    $\quad \Assn{w}{w+w'};$
    $\Assn{(y,u,l)}{\mathsf{head}(w)};$
    $\Rand{p[x]}{{\{0,1\}}^n};$
    $\Assn{a[\epsilon]}{(x,v,p[x])::a[\epsilon]}$
  \end{lstlisting}
  \end{tabular}
  \end{center}
\end{subfigure}\\\bigskip
  \begin{subfigure}[t]{0.5\textwidth}
  \begin{center}
  \underline{Flush operation}
  \begin{tabular}{c}
  \Suppressnumber%
  \begin{lstlisting}[firstnumber=1,numbersep=-3pt]
    $\mbox{}\hspace{-1cm} \mathsf{flush}=$|\Reactivatenumber|
    $\Rand{l}{{\{0,1}\}^n};$
    ${\textcolor{gray}{\Assn{\ell[c].2}{l}}};$
    $\Assn{w}{[]};$    
    $\DFor{j}{1}{n}{}$
    $\quad \Assn{(a[l[1, \ldots, j]],w)}{\mathsf{split}_{\{
    (-,-,l') \mid \mathsf{lcp}(l,l') = l[1, \ldots, i]\}} (a[l[1, \ldots, j]] + w)}$
  \end{lstlisting}
  \end{tabular}
  \end{center}   
  \end{subfigure}%
  \begin{subfigure}[t]{0.47\textwidth}
  \begin{center}
  \underline{Compilation of programs}
  \begin{tabular}{c}
  \Suppressnumber%
  \begin{lstlisting}[firstnumber=1,numbersep=-3pt]
    $\mbox{}\hspace{-1cm} \mathcal{C}(i;P) = \Seq{\mathcal{C}(i)}{\Seq{\mathsf{flush}}{\Seq{\Assn{c}{c + 1}}{\mathcal{C}(P)}}}$
    $\mbox{}\hspace{-1cm} \mathcal{C}(\epsilon) \;\;\; = \Skip$|\Reactivatenumber|
  \end{lstlisting}
  \end{tabular}
  \end{center}   
  \end{subfigure}  
 \caption{Compiling programs. Grayed instructions is ghost code used to record leakage.}\label{fig:oram}
\end{figure}
 
\subsubsection{Security of Simple ORAM}
Informally, the leakage of a program is the sequence of internal
memory accesses performed during program execution. Formally, we
introduce the events $\mathsf{r}~a[i]$ and
$\mathsf{w}~a[i]$ for reading and writing address $i$; note that reading
or writing an internal address does not leak any information. Leakage
is then defined as a sequence of events. For instance, the leakage for
instructions $\iread{x}$ or $\iwrite{x}{v}$ is the sequence:
\[
  \mathsf{r}~a[\epsilon],~\mathsf{w}~a[\epsilon],
\mathsf{r}~a[p[x][1]],~\mathsf{w}~a[p[x][1]], \ldots,
\mathsf{r}~a[p[x]],~\mathsf{w}~a[p[x]],~~\mathsf{w}~a[\epsilon].
\]
The leakage of a program is the concatenation of the leakage of its
instructions. Obliviousness states that executing two different programs with
the same number of instructions induce the same leakage, and that the leakage
does not depend on the initial contents of the ORAM\@.

Our proof uses an equivalent definition of obliviousness that is more convenient
for our purposes. Concretely, we encode the leakage of instructions $\iread{x}$
and $\iwrite{x}{v}$ by $p[x]$. Redefining the leakage in this way does not
affect the definition of obliviousness, but the advantage is that we can show
obliviousness by proving that the leakage of a program $P$ of length $k$ is
uniform over bitstrings of length $2 \cdot n\cdot k$.\footnote{%
  In contrast, using the straightforward definition of leakage would require
proving that a block of memory accesses is uniform over paths from root to leaf.
Although this can also be done in our logic, the proof becomes more cumbersome.}
Our compilation adds ghost code to record this leakage in an array $\ell$
indexed by the program counter $c$, which tracks the index of the current
instruction. Concretely, we add an assignment $\Assn{\ell[c].1}{p[x]}$ at the
beginning of each instruction (line 1 in both \textsf{read} and \textsf{write}
compilation), and an assignment $\Assn{\ell[c].2}{l}$ in corresponding
\textsf{flush} operations (line 2 in \textsf{flush}).

Under our encoding, proving security of the ORAM scheme reduces to showing that
for every program $P$ of length $k$ in our core language, the leakages
$(\ell[1].1, \ell[1].2),\ldots,(\ell[k].1, \ell[k].2)$ are independently and
uniformly distributed in the post-condition of the compiled version of $P$.

We sketch how to formalize this property as a uniformity property in our logic;
details are in \fullref{app:examples}. The overall strategy is to show that
after each instruction-flush pair, the entries of the position map and the
leakage are uniform and mutually independent. We establish two judgments
\begin{gather*}
  \psl
  {\Phi(j)}
  {\Seq{\mathcal{C}(\mathsf{read}(x))}{\Seq{\mathsf{flush}}{\Assn{c}{c + 1}}}}
  {\Phi(j + 1)}
  \\
  \psl
  {\Phi(j)}
  {\Seq{\mathcal{C}(\mathsf{write}(x, v))}{\Seq{\mathsf{flush}}{\Assn{c}{c + 1}}}}
  {\Phi(j + 1)}
\end{gather*}
for every $j$, where the invariant is defined to be
\[
  \Phi(j) \triangleq c = j
  \land \bigsep_{\beta \in [1, c)} \Unif[(\ell[\beta].1, \ell[\beta].2)]
  \sepand \bigsep_{\alpha \in \mathcal{X}} \Unif[p[\alpha]] .
\]
Starting from the pre-condition $\Phi_1 \triangleq \Phi(1)$, which asserts that
the position map is initialized uniformly and independently, repeatedly applying
\rname{Seqn} establishes
\[
  \psl{\Phi(1)}{\mathcal{C}(P)}{\Phi(k + 1)}.
\]
The post-condition implies our desired assertion:
\[
  \Psi \triangleq \bigsep_{\beta \in [1, k]} \Unif[(\ell[\beta].1, \ell[\beta].2)],
\]
which says that the distribution of instruction-flush leakage pairs is uniform
and independent.

\paragraph*{Discussion.}
The ORAM scheme by \citet{iacr/ChungP13a} has several additional wrinkles.
First, they define a recursive ORAM, which uses a hierarchy of trees to reduce
the internal memory; this construction can be encoded in our language, and we
conjecture that our proof can be extended to this more complex setting.
Moreover, \citet{iacr/ChungP13a} initialize the position map lazily, i.e.\, for
every instruction $\iwrite{x}{v}$, the ORAM checks whether the position map for
$x$ is already defined, and samples a fresh bitstring and extends the position
map otherwise. This lazy version can be modeled using conditionals, and our
proof can be adapted to this variant.

\citet{iacr/ChungP13a} also assume that buckets have a maximal size. As a
consequence, execution may fail if buckets overflow. Formalizing this variant
requires care. One option would be to prove that the leakage trace is uniform
conditioned on execution not failing, and that the probability of failing is
small. However, dealing with conditional uniformity is challenging. A better
solution may be to compare the distributions induced by executing the program
with finite buckets and with infinite buckets, showing that these distributions
are close.

\section{Related Work}%
\label{sec:rw}

The intersection of programming languages and security is broad; we limit our
attention to the most directly relevant work.

\paragraph*{Probabilistic independence.}
Our logic is the first program logic where probabilistic independence is the
central concept, but previous systems have also touched on
independence. \citet{DBLP:journals/pacmpl/DaraisLSW20}
define a type and effect system for proving properties of
probabilistic computations. Their effect system is based on a new
notion of probabilistic region, which they use to track probabilistic
dependencies. Their type system uses affine typing to ensure that
random variables are used at most once. They show the expressiveness of
their type and effect system with examples of simple ORAM and
tree-based ORAM\@. Although the two approaches have some similarities,
it seems challenging to compare the expressiveness of their type and
effect system and of our logic. One advantage of our logic is that it
admits an intuitive interpretation based on bunched logics.

\citet{BartheGZ09} define probabilistic Relational Hoare Logic (\textsc{pRHL}),
a program logic for proving relational specifications of probabilistic programs.
Their logic provides a flexible framework for proving information flow
properties of programs. These approaches are able to deal with Private
Information Retrieval, Multi-Party Computation, but are otherwise incomparable
to ours. In particular, it seems difficult to use their approaches for proving
security of ORAM in \textsc{pRHL}, without using an additional proof technique
called Eager/Lazy Sampling. On the other hand, they can prove that von Neumann's
trick, an algorithm to simulate a fair coin using a biased coin, yields a
uniform distribution, which appears out of reach of our current proof system.
\citet{BartheEGHS17} also show how to use this logic to prove uniformity and
independence for probabilistic programs, but these assertions can only be
established at the end of the program.

\SYSTEM{} is also related to \ELLORA, a program logic for probabilistic
programs~\citep{BEGGHS16}. \ELLORA{} works with a more standard assertion logic
based on first order logic, and allows assertions to directly describe
probabilities of events. This expressivity means that it is possible to reason
about independence as a defined assertion. \citet{BEGGHS16} propose an
``independence logic'' as a subsystem, but the rules are limited (e.g., it is
not possible to reason about probabilistic control flow).  In contrast,
probabilistic independence in \SYSTEM{} is handled implicitly by means of a
substructural logic. While this kind of logic is a bit exotic, we find that it
makes it possible to represent independence assertions more compactly and
integrate with mathematical axioms more smoothly.

In more specialized contexts, \citet{DBLP:conf/csfw/Smith03}
and \citet{DBLP:journals/iacr/HoangKM15} develop type systems for proving
computational security of modes of operation and authenticated encryption
schemes. Their type system enforces a strong invariant probabilistic
independence between different expressions, although the type system in itself
does not feature any specific judgment for probabilistic independence. It would
be very interesting to understand whether their results could be emulated and
generalized in our program logic. In a similar spirit,
\citet{DBLP:conf/eurocrypt/BartheBDFGS15} develop a proof system for proving
that programs are protected against power side-channels. Their proof system
makes an implicit but critical use of independence.

On the more foundational side, \citet{DBLP:conferences/lics/AckermanAFRR19}
study computability issues for (conditional) independence. Their work is
partially motivated by exchangeable sequences, which are closely related to
independence. Language-based investigations of exchangeable sequences can be
found for instance in~\citet{DBLP:conf/icalp/StatonSYAF018}. It would be
interesting to investigate how to reason about exchangeable sequences in our
logic.

\paragraph*{Separation logics for probabilistic programs.}
There have been two recent proposals for probabilistic separation logics.
\citet{DBLP:journals/pacmpl/BatzKKMN19} developed a logic \QSL{} for reasoning
about probabilistic, heap-manipulating programs. There, the connectives in BI
are interpreted as acting on \emph{expectations}, real-valued analogs of state
predicates. \citet{DBLP:journals/pacmpl/TassarottiH19} have also developed a
relational separation program logic for reasoning about concurrent probabilistic
programs. Both of these logics leverage standard notions of separation, from heap
separation logic and concurrent separation logic, respectively.

\paragraph*{BI and separation logics.}
Our work builds on fruitful lines of research on bunched implications and
separation logic. On the bunched implications side, our model uses the resource
interpretation of BI~\citep{DBLP:journals/tcs/PymOY04}; readers should consult
\citet{PymMono} or \citet{DochertyThesis} for more information. The idea of
using separation to model probabilistic independence has been considered
before~\citep{lozes2010},\footnote{%
  Peter O'Hearn and David Pym, personal communication.}
but we are not aware of concrete results in this area. From a different point of
view, \citet{DBLP:journals/entcs/Simpson18} develops categorical structures for
independence and conditional independence, encompassing independence in heaps,
nominal sets, and probability distributions. On the separation logic side, by
varying the notion of separation our work is another instance of separation
logic, alongside heap separation logic~\citep{OhRY01,IOh01} and concurrent
separation logic~\citep{DBLP:journals/tcs/OHearn07,DBLP:journals/tcs/Brookes07}.
These areas are too vast to survey here; the draft notes by~\citet{reynolds-SL}
are a good place to start.

\section{Conclusion and Future Directions}%
\label{sec:conclusion}

We have presented a novel separation logic for probabilistic programs, using a
probabilistic variant of the logic BI where separation models probabilistic
independence. Proofs in the separation logic reason in terms of higher-level
properties like independence and uniformity, and we have demonstrated our logic
to prove two forms of cryptographic security for a number of interesting
protocols. We see many possible directions for interesting future work.


\paragraph*{Completeness and decidability.}
Our logic is not complete: there are semantically valid judgments that are not
provable from our proof rules. There are several sources of incompleteness.
First, the proof rules for randomized conditionals is incomplete: it only allows
parts of the pre-condition that are independent of all variables in the guard
expression to be carried into the branches, while a finer analysis could allow
more general pre-conditions to be preserved or modified in a controlled way.
Another source of incompleteness is the interplay between uniformity,
independence, and the equational theory of the expression language; even with
just the xor operator, it is not clear how to give a complete axiomatization.

On the positive side, our logic may be relatively complete under the following
provisos: expressions are variables (i.e., the expression language has no
operators); programs are straightline code (i.e., sequences of assignments); and
assertions are regular or separating conjunctions of uniformity and equality
predicates. In this specific case, it could be possible to reflect a complete
dependency analysis into the program logic; since entailment between our
restricted set of assertions is decidable, this fragment of the logic may also
be decidable.

Decidability of entailment for larger classes of assertions is also an
interesting open problem.  The formulas in our logic bear a superficial
resemblance to the ``pointer logic'' underlying heap separation logic.
Decidability for some restricted fragments follows from a small model
property~\citep{Yang:2001:LRS:933728}. Unlike heap models of BI, our
probabilistic model works with a fixed collection of locations; however, our
setting has probabilistic correlations. We conjecture that fragments of our
logic may also enjoy a small model property, perhaps by tracking which subsets
of variables are mutually independent.


\paragraph*{Enriching the assertion logic.}
The assertion logic we have presented is based on intuitionistic, propositional
{BI}. We have found this logic to be convenient to work with, but other choices
are certainly possible. One natural alternative is to work with a classical
logic instead of an intuitionistic one; in standard separation logic, a
classical logic supports a useful, backwards style of reasoning through the
``magic wand'' connective. In the probabilistic setting, a classical logic runs
into trouble because probabilistic separation seems too strong---we cannot
freely assume that a variable is independent of the rest of the random
variables. However, it may be possible to weaken the notion of separation to
allow randomness to be shared in tightly controlled ways; we are currently
investigating a non-commutative version of BI for this purpose.

There are also natural extensions to our intuitionistic logic. Developing a
probabilistic model of predicate BI~\citep{DBLP:conf/lics/Pym99} would allow
substantially richer assertions. In particular, the lack of existential
quantifiers in our logic complicates our proofs and seems to be an obstacle to
defining a strongest post-condition calculus. Extending the logic to support
reasoning about conditioning would also help make the proof rules more precise.


\paragraph*{Supporting quantitative reasoning.}
While the assertions in our logic describe probability distributions, our logic
notably does not support quantitative reasoning: it is not possible to describe
the probability of an event, or the expected value of a function. This stands in
sharp contrast to other deductive techniques for probabilistic programs, such as
\PPDL~\citep{Kozen:1985} and \PGCL~\citep{Morgan:1996}. Incorporating some of
these tools for reasoning about numeric probabilities could extend the reach of
our logic. One possibility is to make an approximate version of the logic with
judgments of the form $\vdash_\epsilon \psl{\phi}{c}{\psi}$, stating that the
output distribution is at distance at most $\epsilon$ of a distribution
satisfying $\psi$. Such a logic could be obtained by combining ideas of
\SYSTEM{} with the union bound logic of~\citet{BGGHS16b}, and could be used to
reason about more advanced versions of our examples and further examples from
the cryptographic literature (e.g., the PRF/PRP Switching Lemma of
\citet{ImpagliazzoR88}).

 
\paragraph*{Modeling more advanced properties and cryptographic constructions.}
We have focused on basic, information-theoretic security properties from
cryptography in this paper. It would be interesting to understand whether our
logic can be used to capture other properties (e.g., active security).
Similarly, it would be interesting to explore potential applications of our
logic to other constructions, including more complex variants of the
constructions we have considered, e.g., Tree
ORAM~\citep{DBLP:conf/asiacrypt/ShiCSL11,DBLP:conf/pet/GentryGHJ0W13}, Path
ORAM~\citep{DBLP:conf/ccs/StefanovDSFRYD13}, Multi-Server
ORAM~\citep{DBLP:conf/asiacrypt/ChanKNPS18} as well as other constructions, such
as history independent data
structures~\citep{DBLP:conf/stoc/Micciancio97,DBLP:conf/stoc/NaorT01,DBLP:conf/ccs/WangNLCSSH14}.
More speculatively, it would be interesting to understand whether our logic
could be used for reasoning about computational security, or approximate notions
of independence and uniformity.

\begin{acks}
  We thank the anonymous reviewers and our shepherd Ohad Kammar for their close
  reading and useful suggestions. The present work was sparked during a workshop
  at McGill University's Bellairs Research Institute. This work was also
  partially supported by \grantsponsor{ONR}{Office of Naval
  Research}{https://www.onr.navy.mil/} under projects
  \grantnum{ONR}{N00014-12-1-0914}, \grantnum{ONR}{N00014-15-1-2750}, and
  \grantnum{ONR}{N00014-19-1-2292}, the University of Wisconsin, a Facebook TAV
  grant, an NSF Graduate Research Fellowship, and the Max Planck Institute for
  Software-Systems for hosting some of the authors.
\end{acks}

\bibliographystyle{ACM-Reference-Format}
\bibliography{header,main}

\iffull
\appendix

\section{Alternative Proof Rule for Conditionals}%
\label{app:CM}

When the branches of a conditional may modify the guard, rule \rname{RCond} does
not apply. We consider an alternative version:
\[
  \inferrule*[Left=RCondCM]
  { \vdash \psl{\phi \sepand b \sim \ktt}{c}{\psi}
    \\ \vdash \psl{\phi \sepand b \sim \kff}{c'}{\psi}
  \\ \psi \in \CM }
  { \vdash \psl{\phi \sepand \Dist[b]}{\RCond{b}{c}{c'}}{\psi} }
\]
In this case we are not able to show that the guard remains independent of the
post-condition---since it may have been modified---but we are still able to show
that branch post-conditions are preserved. In fact, we may relax the
side-condition on the branch post-condition.

\begin{definition}
  A formula $\phi$ is \emph{closed under mixtures} ($\CM$) if whenever $\mu,
  \mu'$ have the same domain and $(\sigma, \mu) \models \phi$ and $(\sigma,
  \mu') \models \phi$, then $\phi$ is preserved under convex combinations: for every
  $\rho \in [0, 1]$, we have:
  \[
    \dconv{\rho}{(\sigma, \mu)}{(\sigma, \mu')} \models \phi
  \]
\end{definition}

Intuitively, the final distribution after the conditional is a mixture of two
output distributions, one from each branch. Each component distribution
satisfies $\psi$, but the $\CM$ condition is needed to ensure that the mixture
also satisfies $\psi$. The following syntactic conditions ensure $\CM$.

\begin{lemma}\label{lem:CM}
The following assertions are $\CM$, where $\eta$ is $\SP$:
\begin{align*}
  \gamma ::= p_d
    \mid \gamma \land \gamma'
    \mid \eta \sepand \gamma
\end{align*}
\end{lemma}
\begin{proof}
  First, we can show $\models \gamma \to \Dist[\FV(\gamma)]$ by induction on
  $\gamma$. The main lemma then follows by induction on $\gamma$. The only
  interesting case is the last one, when $\gamma = \eta \sepand \gamma'$
  where $\eta$ is $\SP$ and $\gamma'$ is $\CM$.

  Suppose that $(\sigma, \mu) \models \eta \sepand \gamma'$ and $(\sigma, \mu')
  \models \eta \sepand \gamma'$. By validity there are $(\sigma, \nu_1) \models
  \eta$ and $(\sigma, \nu_2) \models \gamma'$ separate such that $(\sigma,
  \nu_1) \circ (\sigma, \nu_2) \sqsubseteq (\sigma, \mu)$, and $(\sigma, \nu_1')
  \models \eta$ and $(\sigma, \nu_2') \models \gamma'$ separate such that
  $(\sigma, \nu_1') \circ (\sigma, \nu_2') \sqsubseteq (\sigma, \mu')$; since
  $\eta$ is $\SP$ we may assume that $\nu_1 = \nu_1'$, and by restriction we may
  assume that $\dom(\nu_2) = \dom(\nu_2') = \FV(\gamma')$ since $\models \gamma'
  \to \Dist[\FV(\gamma')]$. So for any $\rho \in [0, 1]$, we have:
  \begin{align*}
    (\dconv{\rho}{(\sigma, \nu_1)}{(\sigma, \nu_1')}) \circ (\dconv{\rho}{(\sigma, \nu_2)}{(\sigma, \nu_2')})
    &= (\sigma, \nu_1) \circ (\dconv{\rho}{(\sigma, \nu_2)}{(\sigma, \nu_2')})
    \\
    &= (\dconv{\rho}{((\sigma, \nu_1) \circ (\sigma, \nu_2))}{((\sigma, \nu_1) \circ (\sigma, \nu_2'))})
    \\
    &= (\dconv{\rho}{((\sigma, \nu_1) \circ (\sigma, \nu_2))}{((\sigma, \nu_1') \circ (\sigma, \nu_2'))})
    \\
    &\sqsubseteq \dconv{\rho}{(\sigma, \mu)}{(\sigma, \mu')} .
  \end{align*}
  We can conclude, since $\dconv{\rho}{(\sigma, \nu_1)}{(\sigma, \nu_1')} =
  (\sigma, \nu_1) \models \eta$ and $\dconv{\rho}{(\sigma, \nu_2)}{(\sigma,
  \nu_2')} = (\sigma, \nu_2) \models \gamma'$ by induction on $\gamma'$.
\end{proof}

\begin{example}[Non-$\CM$ assertions]
  A simple example of an assertion that is not covered by \cref{lem:CM} is $\phi
  \triangleq \Dist[x] \sepand \Dist[y]$, where $x, y \in \RVar$ are randomized
  variables. In fact, $\phi$ is not $\CM$. To see why, suppose that $x$ and $y$
  are both boolean and consider the distributions $\mu_1 \triangleq \delta_{(x
    \mapsto \ktt, y \mapsto \ktt)}$ and $\mu_2 \triangleq \delta_{(x \mapsto \kff,
  y \mapsto \kff)}$. Then $\phi$ holds in $\mu_1$ and $\mu_2$---in these
  distributions $x$ and $y$ are deterministic, hence independent---but $\phi$ does
  not hold in $\dconv{p}{\mu_1}{\mu_2}$ for any $p \in (0, 1)$.

  Indeed, allowing $\phi$ as a post-condition in \rname{RCondCM} would not be
  sound. Consider the following program:
  \[
    c \triangleq \RCond{x}{\Assn{y}{x}}{\Assn{y}{x}} .
  \]
  Then $\phi$ would be a sound post-condition for each branch, since $x$ and $y$
  are deterministic and hence independent. But $c$ is semantically equal to
  $\Assn{y}{x}$, and $\phi$ is clearly not a sound post-condition.
\end{example}

All $\SP$ assertions are $\CM$, but some $\CM$ assertions are not $\SP$.

\begin{example}[Non-$\SP$ assertions]
  The formula $\phi = \Unif[x] \land \Unif[y]$ is $\CM$ by \cref{lem:CM}, but
  not $\SP$: the following programs have $\phi$ as a post-condition, but have
  incomparable output distributions.
  \begin{align*}
    c_1 &\triangleq \Rand{x}{\Unif_\mathbb{B}}; \Assn{y}{x} \\
    c_2 &\triangleq \Rand{x}{\Unif_\mathbb{B}}; \Assn{y}{\neg x}
  \end{align*}
\end{example}

\section{Omitted Proofs}%
\label{app:proofs}

\LEMrestriction*
\begin{proof}
  The reverse direction follows by the Kripke monotonicity. The forward
  direction follows by induction on $\phi$.
  \begin{itemize}
    \item $\phi \equiv \top, \bot$, and atomic propositions $p$. Trivial.
    \item $\phi \equiv \phi_1 \land \phi_2$. By induction, we have
      \[
        (\sigma, \pi_{\FV(\phi_1)}(\mu) ) \models \phi_1
        \quad\text{and}\quad (\sigma, \pi_{\FV(\phi_2)}(\mu) ) \models \phi_2 .
      \]
      By Kripke monotonicity, we have
      \[
        (\sigma, \pi_{\FV(\phi_1, \phi_2)}(\mu)) \models \phi_1
        \quad\text{and}\quad
        (\sigma, \pi_{\FV(\phi_1, \phi_2)}(\mu)) \models \phi_2
      \]
      so $(\sigma, \pi_{\FV(\phi_1 \land \phi_2)}(\mu)) \models \phi_1 \land \phi_2$.
    \item $\phi \equiv \phi_1 \lor \phi_2$. By induction, we have $(\sigma,
      \pi_{\FV(\phi_i)}(\mu) ) \models \phi_i$ for $i = 1$ or $i = 2$.  By Kripke
      monotonicity, we have $(\sigma, \pi_{\FV(\phi_1, \phi_2)}(\mu)) \models
      \phi_i$ so $(\sigma, \pi_{\FV(\phi_1 \land \phi_2)}(\mu)) \models \phi_1
      \lor \phi_2$.
    \item $\phi \equiv \phi_1 \to \phi_2$. Take any $(\sigma', \mu') \sqsupseteq
      (\sigma, \pi_{\FV(\phi_1, \phi_2)}(\mu))$ such that $(\sigma', \mu')
      \models \phi_1$.  There exists a distribution $\mu''$ such that
      $\dom(\mu'') = \dom(\mu) \cup \dom(\mu')$, and $\pi_{\dom(\mu)}(\mu'') =
      \mu$ and $\pi_{\dom(\mu')}(\mu'') = \mu'$. In particular, $(\sigma',
      \mu'') \sqsupseteq (\sigma, \mu)$. By Kripke monotonicity, we have
      $(\sigma', \mu'') \models \phi_1$ and by validity, we have $(\sigma',
      \mu'') \models \phi_2$. By induction, $(\sigma', \pi_{\FV(\phi_2)}(\mu''))
      \models \phi_2$. Since $(\sigma', \pi_{\FV(\phi_2)}(\mu'')) \sqsubseteq
      (\sigma', \mu')$, Kripke monotonicity gives $(\sigma', \mu') \models
      \phi_2$. So, $(\sigma, \pi_{\FV(\phi_1 \to \phi_2)}(\mu)) \models \phi_1
      \to \phi_2$ as desired.
    \item $\phi \equiv \phi_1 \sepand \phi_2$. There exists $(\sigma_1, \mu_1)$
      and $(\sigma_2, \mu_2)$ with $(\sigma_1, \mu_1) \circ (\sigma_2, \mu_2)
      \sqsubseteq (\sigma, \mu)$ and $(\sigma_1, \mu_1) \models \phi_1$ and
      $(\sigma_2, \mu_2) \models \phi_2$. By induction, we have $(\sigma_1,
      \pi_{\FV(\phi_1)}(\mu_1)) \models \phi_1$ and $(\sigma_2,
      \pi_{\FV(\phi_2)}(\mu_2)) \models \phi_2$. By Kripke monotonicity, we have
      $(\sigma_1, \pi_{\FV(\phi_1 \sepand \phi_2)}(\mu_1)) \models \phi_1$ and
      $(\sigma_2, \pi_{\FV(\phi_1 \sepand \phi_2)}(\mu_2)) \models \phi_2$. Now,
      it is not hard to show that since $(\sigma_1, \mu_1) \circ (\sigma_2,
      \mu_2)$ is defined, $(\sigma_1, \pi_{\FV(\phi_1 \sepand \phi_2)}(\mu_1))
      \circ (\sigma_2, \pi_{\FV(\phi_1 \sepand \phi_2)}(\mu_2)) \sqsubseteq
      (\sigma, \pi_{\FV(\phi_1 \sepand \phi_2)}(\mu))$ is defined as well. So,
      $(\sigma, \pi_{\FV(\phi_1 \sepand \phi_2)}(\mu)) \models \phi_1 \sepand
      \phi_2$ as desired.
    \item $\phi \equiv \phi_1 \sepimp \phi_2$. Take any $(\sigma', \mu')$ such
      that $(\sigma', \mu') \circ (\sigma, \pi_{\FV(\phi_1 \sepimp \phi_2)}(\mu))
      \downarrow$ and $(\sigma', \mu') \models \phi_1$. If $(\sigma', \mu')
      \circ (\sigma, \mu)\downarrow$, then $(\sigma', \mu') \circ (\sigma, \mu)
      \models \phi_2$ and by induction, $(\sigma', \pi_{\FV(\phi_1 \sepimp
      \phi_2)}(\mu')) \circ (\sigma, \pi_{\FV(\phi_1 \sepimp \phi_2)}(\mu))
      \models \phi_2$. Kripke monotonicity gives $(\sigma', \mu') \circ (\sigma,
      \pi_{\FV(\phi_1 \sepimp \phi_2)}(\mu)) \models \phi_2$.

      Otherwise, suppose that $(\sigma', \mu') \circ (\sigma, \mu)$ is not
      defined. Since $(\sigma', \mu') \circ (\sigma, \pi_{\FV(\phi_1 \sepimp
      \phi_2)}(\mu)) \downarrow$, it must be the case that $\emptyset \neq
      \dom(\mu') \cap \dom(\mu) \subseteq \RVar \setminus \FV(\phi_1 \sepimp
      \phi_2)$. Accordingly, $(\sigma', \pi_{\FV(\phi_1)}(\mu')) \circ (\sigma,
      \mu)\downarrow$. By induction, $(\sigma', \pi_{\FV(\phi_1)}(\mu')) \models
      \phi_1$ and so $(\sigma', \pi_{\FV(\phi_1)}(\mu')) \circ (\sigma, \mu)
      \models \phi_2$. By induction again, $(\sigma', \pi_{\FV(\phi_1) \cap
      \FV(\phi_2)}(\mu')) \circ (\sigma, \pi_{\FV(\phi_2)}(\mu)) \models \phi_2$.
      By Kripke monotonicity and the fact that the extension is defined, we have
      $(\sigma', \mu') \circ (\sigma, \pi_{\FV(\phi_1 \sepimp \phi_2)}(\mu))
      \models \phi_2$. So, $(\sigma, \pi_{\FV(\phi_1 \sepimp \phi_2)}(\mu))
      \models \phi_1 \sepimp \phi_2$ as desired.
    \end{itemize}
\end{proof}

\LEMextrusion*
\begin{proof}
  Let $(\sigma, \mu) \models (\phi \sepand \psi) \land \eta$. By validity of the
  first conjunct, there exists separate $(\sigma_1, \mu_1) \models \phi$ and
  $(\sigma_2, \mu_2) \models \psi$. Since $\models \phi \to \Dist[\FV(\eta)]$, we
  have $\FV(\eta) \subseteq \dom(\mu_1)$. By restriction
  (\cref{lem:bi-restriction}) and the fact that $\eta$ is valid in $(\sigma,
  \mu)$, we have $(\sigma_1, \mu_1) \models \eta$. Thus $(\sigma, \mu) \models
  (\phi \land \eta) \sepand \psi$ and so $\models (\phi \sepand \psi) \land \eta
  \to (\phi \land \eta) \sepand \psi$, as desired.
\end{proof}

\LEMsim*
\begin{proof}
  Almost immediate from the definitions; we show (S3). Suppose that $(\sigma,
  \mu)$ is a configuration, and let $(\sigma', \mu') \sqsupseteq (\sigma, \mu)$
  be any larger configuration. If $(\sigma', \mu') \models e_r \sim e_r'$, then
  then free variables of $e_r$ and $e_r'$ are contained in $\dom(\mu')$, and so
  are the free variables of $e_r''$. Now for all $m \in \supp(\mu')$, we have
  $\denot{e_r}(\sigma', m) = \denot{e_r'}(\sigma', m) = \denot{e_r''}(\sigma',
  m)$, where the last equality follows from $\models_E e_r' = e_r''$. Hence
  $(\sigma', \mu') \models e_r \sim e_r''$, as desired.
\end{proof}

\LEMunif*
\begin{proof}
  Almost immediate from definitions; axiom (U3) follows from the fact that the
  uniform distribution is preserved under bijections of its domain.
\end{proof}

\THMsoundness*
\begin{proof}
  By induction on the derivation. Let $(\sigma, \mu)$ satisfy the pre-condition
  of the conclusion. 
  \begin{description}
    \item[DAssn.] By induction on $\psi$.
    \item[Skip.] Trivial.
    \item[Seqn.] By induction hypothesis.
    \item[DCond.] By induction hypothesis and case analysis.
    \item[DLoop.] Since the guard is deterministic and the loop is assumed to
      terminate on all inputs, the number of iterations is a function of the
      deterministic input store and we have:
      \[
        \denot{\DWhile{b}{c}}(\sigma, \mu)
        = \denot{c^{N(\sigma)}}(\sigma, \mu)
        = \denot{c}^{N(\sigma)}(\sigma, \mu)
      \]
      where $c^k \triangleq c \mathbin{;} \cdots \mathbin{;} c$ is the $k$-fold
      sequential composition of $c$. Soundness follows by repeatedly applying
      the induction hypothesis from $c$.
    \item[RAssn.] Trivial.
    \item[RSamp.] Trivial.
    \item[RDCond.] Since $(\sigma, \mu) \models \phi$, either $(\sigma, \mu)
      \models b \sim \ktt$ or $(\sigma, \mu) \models b \sim \kff$.
      Note that exactly one case holds, since $\models b \sim \ktt \to
      \neg (b \sim \kff)$ and vice versa. If $(\sigma, \mu) \models b
      \sim \ktt$ holds, then $(\sigma, \mu) \models \phi \land b \sim
      \ktt$ and since $\denot{\RCond{b}{c}{c'}}(\sigma, \mu) =
      \denot{c}(\sigma, \mu)$, we can conclude by induction. The case $(\sigma,
      \mu) \models b \sim \kff$ is similar.
    \item[RCondCM.] There exist $\mu_1, \mu_2$ such that $\mu_1 \circ \mu_2
      \sqsubseteq \mu$, and $(\sigma, \mu_1) \models \phi$ and $(\sigma, \mu_2)
      \models \Dist[b]$.  Let $\rho$ be the probability $\denot{b =
      \ktt}(\sigma, \mu_2)$. We may assume that $\rho \in (0, 1)$; if
      $\rho$ is equal to zero or one then we can conclude by induction.

      By the semantics of commands, we have
      \[
        \denot{\RCond{b}{c}{c'}}(\sigma, \mu) = \rho \cdot \denot{c}(\sigma,
        \mu_t) + (1 - \rho) \cdot \denot{c'}(\sigma, \mu_f)
      \]
      where $\mu_t$ is the distribution $\mu$ conditioned on $b = \ktt$,
      and $\mu_f$ is the distribution $\mu$ conditioned on $b = \kff$.
      Note that the final deterministic states must be equal to the initial
      deterministic state $\sigma$ in both branches, due to the syntactic
      restriction. 

      Furthermore since $\mu_1$ and $\mu_2$ are independent, we can decompose
      $\mu_1 \circ \mu_{2, t} \sqsubseteq \mu_t$ and $\mu_1 \circ \mu_{2, f}
      \sqsubseteq \mu_f$ such that  $(\sigma, \mu_{2, t}) \models b \sim \ktt$
      and $(\sigma, \mu_{2, f}) \models b \sim \kff$. Since $(\sigma, \mu_1)
      \models \phi$, we know that $(\sigma, \mu_t) \models \phi \sepand b \sim
      \ktt$ and $(\sigma, \mu_f) \models \phi \sepand b \sim \kff$ so the
      induction hypothesis gives:
      \[
        \denot{c}(\sigma, \mu_t) \models \psi
        \text{ and }
        \denot{c'}(\sigma, \mu_f) \models \psi .
      \]
      Since $\psi$ is $\CM$, we can conclude
      \[
        \denot{\RCond{b}{c}{c'}}(\sigma, \mu) 
        = \rho \cdot \denot{c}(\sigma, \mu) + (1 - \rho) \cdot \denot{c'}(\sigma, \mu) 
        \models \psi
      \]
      so the post-condition holds.
    \item[RCond.] The proof goes much like the proof of \rname{RCondCM}; let
      $\mu_t, \mu_f, \rho$ be as before. Recall that by the induction
      hypothesis, we have:
      \[
        \denot{c}(\sigma, \mu_t) \models \psi \sepand b \sim \ktt
        \text{ and }
        \denot{c'}(\sigma, \mu_f) \models \psi \sepand b \sim \kff .
      \]
      Since the top-level command is a randomized conditional, the final
      deterministic state must be the same for both branches; call it $\sigma'$.
      We can decompose the output states into
      \[
        (\sigma', \nu) \circ (\sigma', \nu_t) \sqsubseteq \denot{c}(\sigma, \mu_t)
        \text{ and }
        (\sigma', \nu) \circ (\sigma', \nu_f) \sqsubseteq \denot{c}(\sigma, \mu_f)
      \]
      such that 
      \[
        (\sigma', \nu) \models \psi
        \text{ and }
        (\sigma', \nu_t) \models b \sim \ktt
        \text{ and }
        (\sigma', \nu_f) \models b \sim \kff
      \]
      noting that $\nu$ can be taken to be the same in both branches since $\psi
      \in \SP$; by \cref{lem:bi-restriction}, we may also assume that
      $\dom(\nu_t) = \dom(\nu_f)$. Thus, we have:
      \begin{align*}
        \rho \cdot (\sigma', \nu) \circ (\sigma', \nu_t)  
        + (1 - \rho) \cdot (\sigma', \nu) \circ (\sigma', \nu_f) 
        &= (\sigma', \rho \cdot (\nu \otimes \nu_t) + (1 - \rho) \cdot (\nu \otimes \nu_f)) \\
        &= (\sigma', \nu \otimes (\dconv{\rho}{\nu_t}{\nu_f})) \\
        &= (\sigma', \nu) \circ (\sigma', (\dconv{\rho}{\nu_t}{\nu_f})) \\
        &\sqsubseteq \denot{\RCond{b}{c}{c'}}(\sigma, \mu) ,
      \end{align*}
      and we can conclude since $(\sigma', \nu) \models \psi$ and $(\sigma',
      (\dconv{\rho}{\nu_t}{\nu_f})) \models \Dist[b]$.
    \item[Weak.] By induction hypothesis and semantics of implication.
    \item[True.] Trivial.
    \item[Conj.] By induction hypothesis and semantics of conjunction.
    \item[Case.] By case analysis.
    \item[RCase.] Essentially the same as \rname{RCond}.
    \item[Const.] The fact that $\denot{c}(\sigma, \mu) \models \psi$ follows by
      induction. To show $\denot{c}(\sigma, \mu) \models \eta$, by the
      restriction property we have $(\sigma, \pi_{\FV(\eta)}(\mu)) \models \eta$
      initially, and since the free variables of $\eta$ are disjoint from the
      modified variables of $c$, we have $(\denot{c}\sigma,
      \pi_{\FV(\eta)}(\denot{c}\mu)) \models \eta$ as well. Thus, there is
      restriction of the output where $\eta$ holds, thus $\denot{c}(\sigma, \mu)
      \models \eta$ as desired.
    \item[Frame.] There exist $\mu_1, \mu_2$ such that $\mu_1 \circ \mu_2
      \sqsubseteq \mu$, and $(\sigma, \mu_1) \models \phi$ and $(\sigma, \mu_2)
      \models \eta$; let $S_1 \triangleq \dom(\mu_1)$, and note that $T \cup
      \RV(c) \subseteq S_1$ by the last side-condition.
      
      By the restriction property we have $(\sigma, \pi_{\FV(\eta)}(\mu_2))
      \models \eta$; let $S_2 \triangleq \dom(\mu_2) \cap \FV(\eta)$ and note
      that $S_1$ and $S_2$ are disjoint. Let $S_3$ be the set of all variables
      not contained in $S_1$ or $S_2$. Since $\WV(c)$ is disjoint from $S_2$ by
      the first side-condition, we must have $\WV(c) \subseteq S_1 \cup S_3$.
      
      By induction, we have $\denot{c}(\sigma, \mu) \models \psi$. The
      restriction property gives $(\denot{c}\sigma,
      \pi_{\FV(\psi)}(\denot{c}\mu)) \models \psi$.

      By the third side-condition, $\RV(c) \subseteq S_1$. By soundness of $\RV$
      and $\WV$, all variables in $\WV(c)$ must be written to before they are
      read and there is a function $F : \RMem[S_1] \to \Dist(\RMem[\WV(c) \cup
      S_1])$ such that:
      \[
        (\denot{c}\sigma, \pi_{\WV(c) \cup S_1}(\denot{c}\mu))
        = (\denot{c}\sigma, \dbind(\mu, m \mapsto F(\pi_{S_1}(m)))) .
      \]
      Since $S_2 \subseteq \FV(\eta)$, variables in $S_2$ are not in $\MV(c)$ by
      the first side-condition, and $S_2$ is disjoint from $\WV(c) \cup S_1$. By
      soundness of $\MV$, we have:
      \[
        (\denot{c}\sigma, \pi_{(\WV(c) \cup S_1) \cup S_2}(\denot{c}\mu))
        = (\denot{c}\sigma, \dbind(\pi_{(\WV(c) \cup S_1) \cup S_2}(\mu),
        (m_1, m_2) \mapsto F(m_1) \otimes \dunit(m_2))) .
      \]
      Since $S_1$ and $S_2$ are independent in $\mu$, we know that $S_1 \cup
      \WV(c)$ and $S_2$ are independent in $\denot{c}(\sigma, \mu)$ as well.
      Hence:
      \[
        (\denot{c}\sigma, \denot{c}\mu)
        \sqsupseteq
        (\denot{c}\sigma, \pi_{S_1 \cup \WV(c)}(\denot{c}\mu))
        \circ (\denot{c}\sigma, \pi_{S_2}(\denot{c}\mu)) .
      \]
      We know that $FV(\psi) \subseteq T \cup \WV(c) \subseteq S_1 \cup \WV(c)$
      so since $\psi$ is valid in $\denot{c}(\sigma, \mu)$, it is valid in the
      first conjunct by the restriction property and the second side-condition.
      Since $\pi_{S_2}(\denot{c}\mu) = \pi_{S_2}(\mu)$, and $\eta$ does not
      depend on modified deterministic variables, $\eta$ is valid in the second
      conjunct. Thus, we can conclude:
      \[
        \denot{c}(\sigma, \mu) \models \psi \sepand \eta .
        \qedhere
      \]
  \end{description}
\end{proof}

The proof of the last case relies on the following useful fact connecting
independence and distribution bind.

\begin{lemma}
  Let $A_1, A_2$ be disjoint and let $B_1, B_2$ be disjoint, and consider
  functions $F_i : A_i \to \Dist(\RMem[B_i])$ for $i = 1, 2$. For any two
  distributions $\mu_i \in \Dist(\RMem[A_i])$, we have:
  \[
    \dbind(\mu_1 \otimes \mu_2, (m_1, m_2) \mapsto F_1(m_1) \otimes F_2(m_2))
    = 
    \dbind(\mu_1, m_1 \mapsto F_1(m_1))
    \otimes 
    \dbind(\mu_2, m_2 \mapsto F_2(m_2)) .
  \]
\end{lemma}
\begin{proof}
  By direct calculation.
\end{proof}

\LEMaxioms*
\begin{proof}
  We prove the slightly more general version with a finite set of expressions $\{
  e_i \}$ each with a single random variable $x_i$, and $\{ x_i \}$ are
  distinct. Let $(\sigma, \mu)$ be any configuration.

  For the first axiom, by validity of the left-hand side the configuration can
  be decomposed into a sequence of independent products: $(\sigma, \mu_1) \circ
  \cdots \circ (\sigma, \mu_n) \sqsubseteq (\sigma, \mu)$ such that $x_i \in
  \dom(\mu_i)$ and $(\sigma, \mu_i) \models \Unif_{S_i}[e_i]$. By the
  restriction property, we may assume that $\dom(\mu_i) = \FV(e) = \{ x_i \}$.
  Now the tuple $(e_1, \dots, e_n)$ is uniform in $(\sigma, \mu_1) \circ \cdots
  \circ (\sigma, \mu_n)$, and so $\Unif_{S_1 \times \cdots \times S_n}[(e_1,
  \dots, e_n)]$ holds in a restriction of $(\sigma, \mu)$. The other direction
  is similar.

  For the second axiom, suppose that
  \[
    (\sigma, \mu) \models \Unif_{\mathbb{Z}_q}[e_1] \sepand \Dist[e_2] \sepand
    \cdots \sepand \Dist[e_n] \land e_0 \sim e_1 + \cdots + e_n \mod q .
  \]
  By validity and restriction, we can again decompose $(\sigma, \mu_1) \circ
  \cdots \circ (\sigma, \mu_n) \sqsubseteq (\sigma, \mu)$ such that $(\sigma,
  \mu_1) \models \Unif_{\mathbb{Z}_q}[e_1]$ and $(\sigma, \mu_i) \models
  \Dist[e_i]$ for $i > 1$, and $\dom(\mu_i) = \FV(e_i) = \{ x_i \}$. Now, $e_1 +
  \cdots + e_n \mod q$ is distributed uniformly in $\mu$, since for any
  realization of $e_2, \dots, e_n$ and any $z \in \mathbb{Z}_q$, there is
  exactly one value of $e_1$ that will make $e_1 + \cdots + e_n = z \mod q$ and
  $e_1$ is uniformly distributed, so each $z$ has equal probability. For the
  same reason, $e_1 + \cdots + e_n$ is independent of the joint distribution of
  $(e_2, \dots, e_n)$ in $\mu$. Thus, we have:
  \[
    (\sigma, \mu) \models \Unif_{\mathbb{Z}_q}[e_0] \sepand \Dist[(e_2, \dots, e_n)] .
  \]
  Since $(\sigma, \mu) \models \Dist[e_2] \sepand \cdots \sepand \Dist[e_n]$,
  extrusion gives
  \[
    (\sigma, \mu) \models \Unif_{\mathbb{Z}_q}[e_0] \sepand \Dist[e_2] \sepand \cdots \sepand \Dist[e_n] 
  \]
  as desired.
\end{proof}

\section{Examples: Additional Details}%
\label{app:examples}

\subsection{Private Information Retrieval}

\subsubsection{Proof of Uniformity}

Starting from the trivial pre-condition $\Phi_1 \triangleq \top$, we
would like to prove the post-condition
\[
  \Psi \triangleq \underbrace{\Unif[q_0]}_{\mathclap{\text{$\mathcal{S}_0$'s view}}} \land \underbrace{\Unif[q_1]}_{\mathclap{\text{$\mathcal{S}_1$'s view}}}.
\]
This says that the views of $\mathcal{S}_0$ and $\mathcal{S}_1$ ($q_0$ and $q_1$, respectively) are
uniformly random bitstrings.

By \rname{RSamp}, adjoining the sampling for  $q_0$ (line 1) gives
\[
  \Unif[q_0].
\]
Since $I$ is a deterministic variable, we can adjoin $\Dist[I]$, giving
\[
  \Unif[q_0] \sepand \Dist[I].
\]
By \rname{RAssn*}, assigning to $q_1$ (line 2) gives
\[
  \Unif[q_0] \sepand \Dist[I] \mathrel{\land} q_1 \sim q_0 \mathrel{\oplus} I.
\]
Next, we can pull out $\Unif[q_0]$ like so
\[
  \Unif[q_0] \mathrel{\land} (\Unif[q_0] \sepand \Dist[I] \mathrel{\land} q_1 \sim q_0 \mathrel{\oplus} I).
\]
and apply the xor axiom (\ref{ax:xor}) to the right conjunct, which gives the
desired post-condition
\[
  \Psi \triangleq \Unif[q_0] \mathrel{\land} \Unif[q_1].
\]
Since $q_0$ and $q_1$ are unmodified in the remainder of the program, we can
preserve $\Psi$ through to the end using \rname{Const} and \rname{True}.

\subsubsection{Proof of Input Independence}

Starting from the pre-condition $\Phi_1 \triangleq \Dist[I]$, we would
like to prove the post-condition
\[
  \Psi \triangleq \underbrace{\Dist[I]}_{\mathclap{\text{Index}}} \sepand \underbrace{\Unif[q_0]}_{\mathclap{\text{$\mathcal{S}_0$'s
  view}}} \mathrel{\land} \underbrace{\Dist[I]}_{\mathclap{\text{Index}}} \sepand \underbrace{\Unif[q_1]}_{\mathclap{\text{$\mathcal{S}_1$'s
  view}}}.
\]
This says that the views of $\mathcal{S}_0$ and $\mathcal{S}_1$ ($q_0$ and $q_1$, respectively) are
independent of secret index $I$.

By \rname{RSamp*}, adjoining the sampling for $q_0$ (line 1) gives
\[
  \Unif[q_0] \sepand \Dist[I].
\]
By \rname{RAssn*}, assigning to $q_1$ (line 2) gives
\[
  \Unif[q_0] \sepand \Dist[I] \mathrel{\land} q_1 \sim q_0 \mathrel{\oplus} I.
\]
Next, we can pull out $\Dist[I] \sepand \Unif[q_0]$ like so
\[
  \Dist[I] \sepand \Unif[q_0] \mathrel{\land}
  (\Unif[q_0] \sepand \Dist[I] \mathrel{\land} q_1 \sim q_0 \mathrel{\oplus} I)
\]
and apply the xor axiom (\ref{ax:xor}) to the right conjunct, giving
\[
  \Dist[I] \sepand \Unif[q_0] \mathrel{\land} \Dist[I] \sepand \Unif[q_1],
\]
which implies the desired post-condition
\[
  \Psi \triangleq \Dist[I] \sepand \Unif[q_0] \mathrel{\land} \Dist[I] \sepand \Unif[q_1].
\]
Since $q_0$ and $q_1$ are unmodified in the remainder of the program, we can
preserve $\Psi$ through to the end using \rname{Const} and \rname{True}.

\subsection{Oblivious Transfer}

\subsubsection{Proof of Uniformity}

Starting from the trivial pre-condition $\Phi_1 \triangleq \top$, we would like
to prove the post-condition
\[
  \Psi \triangleq \underbrace{(\Unif_{k \times
  k}[(r_0,r_1)] \sepand \Unif[e])}_{\mathclap{\text{$\mathcal{S}$'s
  view}}} \mathrel{\land} \underbrace{(\Unif[d] \sepand \Unif_{k \times
  k}[(r_d,f_{1-c})])}_{\mathclap{\text{$\mathcal{R}$'s ``view''}}}.
\]
Starting from the trivial pre-condition $\Phi_1 \triangleq \top$, we would like
to prove the post-condition
\[
  \Psi \triangleq \underbrace{(\Unif_{k \times
  k}[(r_0,r_1)] \sepand \Unif[e])}_{\mathclap{\text{$\mathcal{S}$'s
  view}}} \mathrel{\land} \underbrace{(\Unif[d] \sepand \Unif_{k \times
  k}[(r_d,f_{1-c})])}_{\mathclap{\text{$\mathcal{R}$'s ``view''}}}.
\]
To establish $\mathcal{R}'s$ secrecy, we need to consider the view of
$\mathcal{S}$, which consists of $r_0$, $r_1$, and $e$. For $\mathcal{R}$'s
choice $c$ to be kept secret, it is required that $\Unif_{k \times
k}[(r_0,r_1)] \sepand \Unif[e]$, i.e., $\mathcal{S}$'s combined view is
uniform. Note that it is not enough to establish that the individual components
of $\mathcal{S}$'s view are uniform. To see why, suppose $\mathcal{T}$ also
sends $\mathcal{S}$ the random bit $d$, which reveals $c = e \oplus d$. Although
the individual components of $\mathcal{S}$'s view would indeed be uniform, i.e.,
$\Unif_k[r_0] \mathrel{\land} \Unif_k[r_1] \mathrel{\land} \Unif[e] \mathrel{\land}
\Unif[d]$, $\mathcal{R}$'s secrecy is clearly violated. Thus, the stronger
post-condition is needed to establish that $\mathcal{S}$'s combined view is
uniform.

To establish $\mathcal{S}'s$ (one-sided) secrecy, we need to consider the view
of $\mathcal{R}$, which consists of $d$, $r_d$, and one of $f_0$ or $f_1$. In
particular, the $f_i$ to be considered corresponds to the encryption of the
``wrong'' message, which we assign to the ghost variable $f_{1-c}$ (which is, in
turn, computed using the ghost variable $r_{1-d}$). Similar to $\mathcal{R}$'s
secrecy, it is then required that $\Unif[d] \sepand \Unif_{k \times
k}[(r_d,f_{1-c})]$, i.e., $\mathcal{R}$'s combined view is uniform.

We first show $\mathcal{R}$'s secrecy, followed by $\mathcal{S}$'s secrecy, and
then combine the results using \rname{Conj}.  By \rname{RSamp}
and \rname{RSamp*}, we can adjoin the random samplings for
$r_0,r_1,d$ (lines 1--2), giving
\[
  \Unif_k[r_0] \sepand \Unif_k[r_1] \sepand \Unif[d].
\]
Since the free variables of this formula are unmodified in the conditional (line
3), we can preserve the formula using \rname{Const} and \rname{True}.
Since $c$ is a deterministic variable we can adjoin $\Dist[c]$, giving
\[
  \Dist[c] \sepand \Unif_k[r_0] \sepand \Unif_k[r_1] \sepand \Unif[d].
\]
For the assignment to $e$ (line 4), we start from the local pre-condition
\[
  \Dist[c] \sepand \Unif[d].
\]
By \rname{RAssn*}, assigning to $e$ gives
\[
  (\Dist[c] \sepand \Unif[d]) \mathrel{\land} e \sim c \mathrel{\oplus} d.
\]
Applying the xor axiom (\ref{ax:xor}) leaves
\[
  \Unif[e].
\]
Then, we can frame as follows:
\begin{mathpar}
    \hspace*{0.9cm}\inferrule*[Left=Frame]
    { \vdash \psl{\phi}{c'}{\psi} \quad\,\,\,
      \FV(\eta) \cap \MV(c') = \emptyset \quad\,\,\,
      \FV(\psi) \subseteq \FV(\phi) \cup \WV(c') \quad\,\,\,
    \models \phi \to \Dist[\RV(c')]}
    { \vdash \psl{\underbrace{\Dist[c] \sepand \Unif[d]}_\phi \sepand \underbrace{\Unif_k[r_0] \sepand \Unif_k[r_1]}_\eta}{\underbrace{\Assn{e}{c \oplus
    d}}_{c'}}{\underbrace{\Unif[e]}_\psi \sepand \underbrace{\Unif_k[r_0] \sepand \Unif_k[r_1]}_\eta}
    }.
\end{mathpar}
The post-condition implies
\[
  \Unif_{k \times k}[(r_0,r_1)] \sepand \Unif[e],
\]
which establishes $\mathcal{R}$'s secrecy. Since $r_0$, $r_1$, and $e$ are
unmodified in the remainder of the program, we can preserve $\Unif_{k \times
k}[(r_0,r_1)] \sepand \Unif[e]$ through to the end using \rname{Const}.

Next, we show $\mathcal{S}$'s secrecy. Again, by \rname{RSamp}
and \rname{RSamp*}, we can adjoin the random samplings for $r_0,r_1,d$ (lines
1--2), giving
\[
  \Unif_k[r_0] \sepand \Unif_k[r_1] \sepand \Unif[d].
\]
We go through the conditional (line 3) with \rname{RCond}, which gives
pre-condition
\[
  \Unif_k[r_0] \sepand \Unif_k[r_1] \sepand d = 0 \sim \ktt.
\]
To go through the first assignment, we start from the local pre-condition
\[
  \Unif_k[r_0].
\]
By \rname{RAssn*}, assigning to $r_d$ gives
\[
  \Unif_k[r_0] \mathrel{\land} r_d \sim r_0.
\]
Transferring the distribution law gives
\[
  \Unif_k[r_d].
\]
We can then frame in $\Unif_k[r_1] \sepand d = 0 \sim \ktt$ giving
\[
  \Unif_k[r_d] \sepand \Unif_k[r_1] \sepand d = 0 \sim \ktt.
\]
Going through the second assignment follows similarly, this time starting from
the local pre-condition
\[
  \Unif_k[r_1]
\]
and giving the following post-condition in the true branch
\[
  \Unif_k[r_d] \sepand \Unif_k[r_{1-d}] \sepand d = 0 \sim \ktt.
\]
The false branch yields the same post-condition, which, by \rname{RCond},
brings us to the post-condition
\[
  \Unif_k[r_d] \sepand \Unif_k[r_{1-d}] \sepand \Unif[d].
\]
Since the free variables of this formula are unmodified in lines 4--5, we can
preserve this formula through using \rname{Const} and \rname{True}. Next, we go
through the deterministic conditional (line 6) using \rname{DCond}. In the true
branch, we start with pre-condition
\[
  \Unif_k[r_d] \sepand \Unif_k[r_{1-d}] \sepand \Unif[d] \mathrel{\land} (c=0)
  = \ktt.
\]
Dropping the right conjunct, we can adjoin $\Dist[m_1]$ like so
\[
  \Unif_k[r_d] \sepand \Dist[m_1] \sepand \Unif_k[r_{1-d}] \sepand \Unif[d],
\]
since $m_1$ is a deterministic variable.  We preserve this formula through the
first assignment to $m_c$ using \rname{Const} and \rname{True}, and then go
through the second assignment to $f_{1-c}$ starting from the local pre-condition
\[
  \Dist[m_1] \sepand \Unif_k[r_{1-d}].
\]
Applying \rname{RAssn*} and the xor axiom (\ref{ax:xor}) gives
\[
  \Unif_k[f_{1-c}].
\]
Framing in $\Unif_k[r_d] \sepand \Unif[d]$ gives the following post-condition in
the true branch
\[
  \Unif_k[r_d] \sepand \Unif_k[f_{1-c}] \sepand \Unif[d].
\]
The false branch yields the same post-condition. Then, we can merge
$\Unif_k[r_d] \sepand \Unif_k[f_{1-c}]$ and rearrange like so
\[
    \Unif[d] \sepand \Unif_{k \times k}[(r_d,f_{1-c})].
\]
This establishes $\mathcal{S}$'s secrecy. Combining the formulas for
$\mathcal{R}$'s secrecy and $\mathcal{S}$'s secrecy using \rname{Conj} gives the
desired post-condition.

\subsection{Multi-Party Computation}

\subsubsection{Proof of Uniformity}

Starting from the trivial pre-condition $\Phi_1 \triangleq \top$, we would like
to prove the post-condition
\[
  \Psi \triangleq \bigwedge_{\alpha \in \{2,3\}} \underbrace{\Unif[(r[\alpha].2,
  r[\alpha].3)]}_{\mathclap{\text{$P_1$'s view from $P_\alpha$}}} \mathrel{\land}
  \bigwedge_{\alpha \in \{1,3\}} \underbrace{\Unif[(r[\alpha].1,
  r[\alpha].3)]}_{\mathclap{\text{$P_2$'s view from
  $P_\alpha$}}}  \mathrel{\land}
  \bigwedge_{\alpha \in \{1,2\}} \underbrace{\Unif[(r[\alpha].1, r[\alpha].2)]}_{\mathclap{\text{$P_3$'s view from $P_\alpha$}}}.
\]
This says that each party's view from the other parties is uniform and independent.

To prove this post-condition, we take the following for-loop invariant:
\[
  \bigsep_{\alpha \in [1,
    i)} \Unif[(r[\alpha].1,r[\alpha].2)] \mathrel{\land} \Unif[(r[\alpha].2,r[\alpha].3)] \mathrel{\land} \Unif[(r[\alpha].1,r[\alpha].3)].
\]
By \rname{RSamp*}, adjoining the random samplings for $r[i].1$ and $r[i].2$
(lines 2 and 3) gives
\[
  \Unif[r[i].1] \sepand \Unif[r[i].2] \sepand \bigsep_{\alpha \in [1,
    i)} \Unif[(r[\alpha].1,r[\alpha].2)] \mathrel{\land} \Unif[(r[\alpha].2,r[\alpha].3)] \mathrel{\land} \Unif[(r[\alpha].1,r[\alpha].3)].
\]
To go through the assignment to $r[i].3$ (line 4), we start from the local pre-condition
\[
  \Unif[r[i].1] \sepand \Unif[r[i].2].
\]
By \rname{RAssn*}, assigning to $r[i].3$ gives
\[
  \Unif[r[i].1] \sepand \Unif[r[i].2] \mathrel{\land} r[i].3 \sim x[i] - r[i].1
  - r[i].2 \mod p.
\]
Applying the modular addition axiom (\ref{ax:mod}) gives
\[
  \Unif[r[i].1] \sepand \Unif[r[i].2] \mathrel{\land} \Unif[r[i].2] \mathrel{\sepand} \Unif[r[i].3] \mathrel{\land} \Unif[r[i].1] \mathrel{\sepand} \Unif[r[i].3].
\]
Then, we can merge pairwise independent distributions like so
\[
  \Unif[(r[i].1,r[i].2)] \mathrel{\land} \Unif[(r[i].2,r[i].3)] \mathrel{\land}
  \Unif[(r[i].1,r[i].3)].
\]
Then, we can frame as follows
\begin{mathpar}
    \hspace*{1cm}\inferrule*[Left=Frame]
    { \vdash \psl{\phi}{c}{\psi} \\
      \FV(\eta) \cap \MV(c) = \emptyset \\
      \FV(\psi) \subseteq \FV(\phi) \cup \WV(c) \\
    \models \phi \to \Dist[\RV(c)]}
    { \vdash \Bigr\{ \underbrace{\Unif[r[i].1] \sepand \Unif[r[i].2]}_{\phi} \sepand \eta \Bigl\}\\\\
    c \triangleq \Assn{r[i].3}{x[i] - r[i].1 - r[i].2 \mod p} \\\\
    \Bigl\{ \underbrace{\Unif[(r[i].1,r[i].2)] \mathrel{\land} \Unif[(r[i].2,r[i].3)] \mathrel{\land}
  \Unif[(r[i].1,r[i].3)]}_{\psi} \sepand \eta \Bigr\}
    },
\end{mathpar}
where
\[
  \eta \triangleq \bigsep_{\alpha \in [1,
    i)} \Unif[(r[\alpha].1,r[\alpha].2)] \mathrel{\land} \Unif[(r[\alpha].2,r[\alpha].3)] \mathrel{\land} \Unif[(r[\alpha].1,r[\alpha].3)].
\]
We can reassociate the post-condition like so
\[
  \bigsep_{\alpha \in [1,
    i+1)} \Unif[(r[\alpha].1,r[\alpha].2)] \mathrel{\land} \Unif[(r[\alpha].2,r[\alpha].3)] \mathrel{\land} \Unif[(r[\alpha].1,r[\alpha].3)].
\]
The for-loop post-condition from \rname{DFor} implies
\[
  \bigsep_{\alpha \in
  [1,3]} \Unif[(r[\alpha].1,r[\alpha].2)] \mathrel{\land} \Unif[(r[\alpha].2,r[\alpha].3)] \mathrel{\land} \Unif[(r[\alpha].1,r[\alpha].3)].
\]
After rearranging and dropping terms, this implies the desired post-condition
\[
  \Psi \triangleq \bigwedge_{\alpha \in \{2,3\}} \Unif[(r[\alpha].2,
  r[\alpha].3)] \mathrel{\land} \bigwedge_{\alpha \in \{1,3\}} \Unif[(r[\alpha].1,
  r[\alpha].3)] \mathrel{\land} \bigwedge_{\alpha \in \{1,2\}} \Unif[(r[\alpha].1,
  r[\alpha].2)].
\]
Since the free variables of $\Psi$ are unmodified in the remainder of the
program, we can preserve $\Psi$ through to the end using \rname{Const}
and \rname{True}.

\subsubsection{Proof of Input Independence}

Starting from the pre-condition
\[
  \Phi_1 \triangleq \bigwedge_{\alpha \in [1,3]} \Dist[x[\alpha]],
\]
we would like to prove the post-condition
\begin{align*}
  \Psi \triangleq \bigwedge_{\alpha \in \{2,3\}} \underbrace{\Dist[x[\alpha]]}_{\mathclap{\text{$P_\alpha$'s
  input}}} \sepand \underbrace{\Dist[(r[\alpha].2,
  r[\alpha].3)]}_{\mathclap{\text{$P_1$'s view from
  $P_\alpha$}}} \mathrel{\land} &\bigwedge_{\alpha \in \{1,3\}} \underbrace{\Dist[x[\alpha]]}_{\mathclap{\text{$P_\alpha$'s
  input}}} \underbrace{\sepand \Dist[(r[\alpha].1,
  r[\alpha].3)]}_{\mathclap{\text{$P_2$'s view from
  $P_\alpha$}}} \mathrel{\land}\\ &\bigwedge_{\alpha \in \{1,2\}} \underbrace{\Dist[x[\alpha]]}_{\mathclap{\text{$P_\alpha$'s
  input}}} \sepand \underbrace{\Dist[(r[\alpha].1,
  r[\alpha].2)]}_{\mathclap{\text{$P_3$'s view from $P_\alpha$}}}.
\end{align*}
This says that, for each party, the secret input of each other party is
independent from the view they generate. Throughout, let
\[
  \Phi_2 \triangleq \Dist[x[\alpha]]
  \sepand \Unif[(r[\alpha].1,r[\alpha].2)]
  \mathrel{\land} \Dist[x[\alpha]]
  \sepand \Unif[(r[\alpha].2,r[\alpha].3)]
  \mathrel{\land} \Dist[x[\alpha]]
  \sepand \Unif[(r[\alpha].1,r[\alpha].3)].
\]
First, we take the following for-loop invariant
\[
  \bigwedge_{\alpha \in
  [1,3]} \Dist[x[\alpha]] \mathrel{\land} \bigwedge_{\alpha \in [1,i)} \Phi_2.
\]
To go through the assignments in lines 2--4, we start from the local pre-condition
\[
  \Dist[x[i]].
\]
By \rname{RSamp*}, adjoining the random samplings for $r[i].1$ and $r[i].2$ gives
\[
  \Dist[x[i]] \sepand \Unif[r[i].1] \sepand \Unif[r[i].2].
\]
By rule \rname{RAssn*}, assigning to $r[i].3$ gives
\[
  (\Dist[x[i]] \sepand \Unif[r[i].1] \sepand \Unif[r[i].2]) \mathrel{\land}
  r[i].3 \sim x[i] - r[i].1 - r[i].2 \mod p.
\]
Applying the modular arithmetic axiom (\ref{ax:mod}) gives
\[
  \Dist[x[i]] \sepand \Unif[r[i].1] \sepand \Unif[r[i].2] \mathrel{\land} \Dist[x[i]] \sepand \Unif[r[i].2] \mathrel{\sepand} \Unif[r[i].3] \mathrel{\land} \Dist[x[i]] \sepand \Unif[r[i].1] \mathrel{\sepand} \Unif[r[i].3].
\]
We can then merge the pairwise independent distributions of secret shares like so
\[
  \Phi_3 \triangleq \Dist[x[i]] \sepand \Unif[(r[i].1,r[i].2)] \mathrel{\land} \Dist[x[i]] \sepand \Unif[(r[i].2,r[i].3)] \mathrel{\land} \Dist[x[i]] \sepand \Unif[(r[i].1,r[i].3)].
\]
Then, we can carry in unused conjuncts as follows
\begin{mathpar}
    \inferrule*[Left=Const]
    { \vdash \psl{\phi}{c}{\psi} \\ \FV(\eta) \cap \MV(c) = \emptyset }
    { \vdash \Bigl\{\underbrace{\Dist[x[i]]}_{\phi} \mathrel{\land} \underbrace{\bigwedge_{\alpha \in
    [1,3] \setminus i} \Dist[x[i]] \mathrel{\land} \bigwedge_{\alpha \in
    [1,i)} \Phi_2}_\eta \Bigr\}\\\\
    c \triangleq \Seq{\Rand{r[i].1}{\mathbb{Z}_p}}{\Seq{\Rand{r[i].2}{\mathbb{Z}_p}}{\Assn{r[i].3}{x[i]
    - r[i].1 - r[i].2 \mod p}}}\\\\
    \Bigl\{
    \underbrace{\mathclap{\Phi_3}}_{\psi}
    \mathrel{\land} \underbrace{\bigwedge_{\alpha \in
    [1,3] \setminus i} \Dist[x[i]] \mathrel{\land} \bigwedge_{\alpha \in
    [1,i)} \Phi_2}_\eta \Bigr\}}.
\end{mathpar}
Reassociating the post-condition gives
\[
\bigwedge_{\alpha \in
    [1,3]} \Dist[x[i]] \mathrel{\land} \bigwedge_{\alpha \in [1,i+1)} \Phi_2.
\]
Thus, the for-loop invariant is preserved. By \rname{DFor}, the post-condition of
the conclusion is
\[
\bigwedge_{\alpha \in
    [1,3]} \Dist[x[i]] \mathrel{\land} \bigwedge_{\alpha \in
  [1,4)} \Phi_2,
\]
which implies the desired post-condition $\Psi$
\begin{align*}
  \bigwedge_{\alpha \in \{2,3\}} \Dist[x[\alpha]] \sepand \Dist[(r[\alpha].2,
  r[\alpha].3)] \mathrel{\land} &\bigwedge_{\alpha \in \{1,3\}} \Dist[x[\alpha]] \sepand \Dist[(r[\alpha].1,
  r[\alpha].3)] \mathrel{\land}\\ &\bigwedge_{\alpha \in \{1,2\}} \Dist[x[\alpha]] \sepand \Dist[(r[\alpha].1,
  r[\alpha].2)].
\end{align*}
Because the free variables of $\Psi$ are unmodified in the rest of the program,
we can frame $\Psi$ through to the end using \rname{Const} and \rname{True}.

\subsection{Simple Oblivious RAM}

Starting from the pre-condition
\[
  \Phi_1 \triangleq \bigsep_{\alpha \in \mathcal{X}} \Unif[p[\alpha]],
\]
which says that the position map is initialized uniformly and independently, we
would like to prove the post-condition
\[
  \Psi \triangleq \bigsep_{\beta \in [1, k]} \Unif[(\ell[\beta].1, \ell[\beta].2)],
\]
which says that the product distribution for each instruction-flush leakage pair
is uniform and independent.

Starting from the local pre-condition
\[
  \Unif[p[x]],
\]
the assignment to $\ell[j].1$ (using \rname{RAssn*}) gives
\[
  \Unif[p[x]] \mathrel{\land} \ell[1].1 \sim p[x].
\]
Transferring the distribution law gives
\[
  \Unif[\ell[1].1].
\]
Then, we can frame as follows
\begin{mathpar}
    \hspace*{1cm}\inferrule*[Left=Frame]
    { \vdash \psl{\phi}{c}{\psi} \\
      \FV(\eta) \cap \MV(c) = \emptyset \\
      \FV(\psi) \subseteq \FV(\phi) \cup \WV(c) \\
    \models \phi \to \Dist[\RV(c)]}
    { \vdash \Bigl\{ \underbrace{\Unif[p[x]]}_{\phi} \sepand \underbrace{\bigsep_{\alpha \in \mathcal{X} \setminus
    x} \Unif[p[\alpha]]}_{\eta}\Bigr\}
  ~\underbrace{\Assn{\ell[1].1}{p[x]}}_c
  \Bigl\{ \underbrace{\Unif[\ell[i].1]}_{\psi} \sepand \underbrace{\bigsep_{\alpha \in \mathcal{X} \setminus
  x} \Unif[p[\alpha]]}_{\eta} \Bigr\}}
\end{mathpar}
We can preserve this post-condition through lines 4--8 up until the random
sampling for $p[x]$ using
\rname{Const} and \rname{True}. By \rname{RSamp*}, the random sampling for $p[x]$ gives
\[
  \Unif[\ell[1].1] \sepand \bigsep_{\alpha \in \mathcal{X}} \Unif[p[\alpha]].
\]
Again, we can frame this formula through the last assignment in
the \textsf{read} instruction using \rname{Const} and \rname{True}. For
the \textsf{flush} instruction, we first adjoin the random sampling for $l$
\[
  \Unif[\ell[1].1] \sepand \Unif[l] \sepand \bigsep_{\alpha \in \mathcal{X}} \Unif[p[\alpha]].
\]
Assigning to $\ell[1].2$ gives
\[
  \Unif[\ell[1].1] \sepand \Unif[\ell[1].2] \sepand \bigsep_{\alpha \in \mathcal{X}} \Unif[p[\alpha]].
\]
Then, we can merge $\Unif[\ell[1].1] \sepand \Unif[\ell[1].2]$ like so
\[
  \Unif[(\ell[1].1,\ell[1].2)] \sepand \bigsep_{\alpha \in \mathcal{X}} \Unif[p[\alpha]].
\]
Then, we preserve through rest of the \textsf{flush} instruction
using \rname{Const} and \rname{True}. We proceed through the remaining $k-1$
compiled instructions similarly, giving the desired post-condition
\[
  \Psi \triangleq \bigsep_{\beta \in [1, k]} \Unif[(\ell[\beta].1, \ell[\beta].2)].
\]

\fi
\end{document}